\documentclass[11pt]{article}
\usepackage{amsfonts,amsmath,amsthm,amssymb,graphicx,mathrsfs}

\usepackage[top=2cm,
			bottom=2cm,
			left=2cm,
			right=2cm]{geometry}
\usepackage[colorlinks=true,urlcolor=blue, linkcolor =blue,citecolor=blue]{hyperref}
\usepackage{algorithm,algorithmicx,algpseudocode,multirow}
\usepackage{booktabs}

\theoremstyle{plain}
\newtheorem{theorem}{Theorem}
\newtheorem{dfn}{Definition}
\newtheorem{lemma}{Lemma}
\newtheorem{assume}{Assumption}
\newtheorem{remark}{Remark}
\newtheorem{corollary}{Corollary}

\allowdisplaybreaks

\begin{document}

\title{\bf From bilinear regression to inductive matrix completion: \\
	a quasi-Bayesian analysis}

\author{The Tien Mai }

\date{
\begin{small}
Department of Mathematical Sciences, 
\\
Norwegian University of Science and Technology,
Trondheim, Norway.
\\
Email: the.t.mai@ntnu.no
\end{small}
}

\maketitle

\begin{abstract}
In this paper we study the problem of bilinear regression and we further address the case when the response matrix contains missing data that referred as the problem of inductive matrix completion. We propose a quasi-Bayesian approach first to the problem of bilinear regression where a quasi-likelihood is employed. Then, we adapt this approach to the context of inductive matrix completion. Under a low-rankness assumption and leveraging PAC-Bayes bound technique, we provide statistical properties for our proposed estimators and for the quasi-posteriors. We propose a Langevin Monte Carlo method to approximately compute the proposed estimators. Some numerical studies are conducted to demonstrated our methods. 
\end{abstract}


\section{Introduction}

In this paper we study the bilinear regression model where a linear relationship is assumed between a set multiple response variables and two sets of covariates. Often this model is referred as the growth curve model or generalized multivariate analysis model that is often used for analyzing longitudinal data, see for example \cite{rosen2021bilinear,von2018bilinear,potthoff1964generalized,woolson1980growth,kshirsagar1995growth,jana2017inference}. However, these mentioned works cover only the case in which the response matrix is fully observed. The bilinear regression model with incomplete response is recently introduced and studied as the so-called inductive matrix completion, a generalization of the matrix completion problem, \cite{natarajan2014inductive,zilber2022inductive}. Inductive matrix completion drawn significant interest in different fields such as drug repositioning \cite{zhang2020drimc}, collaborative filtering \cite{hsieh2015pu}, genomics \cite{natarajan2014inductive}.

Under a low-rank constraint on the coefficient matrix, most approaches for bilinear regression and inductive matrix completion are frequentist methods such as maximum likelihood estimation \cite{von2018bilinear} or penalized optimization \cite{natarajan2014inductive}. One the other hand, Bayesian approaches have been recently considered discretely for these problems, for example, the paper 
\cite{jana2019bayesian} proposed a Bayesian approach for bilinear regression where as a Bayesian method was proposed for inductive matrix completion in the work \cite{zhang2020drimc}. However, unlike frequentist approaches, statistical properties of the Bayesian approach for these models have not been carried out yet. 

In this work, we aim at filling these gaps by deriving some theoretical properties of a Bayesian approach for these two problems. More specifically, we propose a quasi-Bayesian approach to the problem of bilinear regression where a quasi-likelihood is used. Then, we further generalize this approach to the problem of inductive matrix completion. Using a spectral scaled Student prior, we show that the posterior mean satisfy a tight oracle inequality proving that our method is adaptive to the rank of the coefficient matrix. Contraction properties of the posteriors are also proved.

Quasi-Bayesian method is an expansion of the Bayesian approach and it is increasingly popular in statistics and machine learning, as the so-called generalized Bayesian inference, \cite{kno2019,bissiri2013general,grunwald2017inconsistency}. This approach allows to relax the generating model assumption where the likelihood can be replaced by a more general notion of risk or a quasi-likelihood. Using PAC-Bayesian technique \cite{McA,STW,catonibook}, we derive theoretical results for our proposed quasi-posteriors.  We refer the reader to \cite{gue2019,TUTO} for recent reviews and introductions to this technique. PAC-Bayes bounds were successfully used before in the context matrix estimation problems, see for example \cite{mai2015,cottet20181,mai2017pseudo,mai2022optimal}. Interestingly, using this technique for inductive matrix completion, we do not need to put any assumption on the distribution of the missing entries in the response matrix. This is in contrast with previous works on matrix completion such as \cite{jain2013provable,candes2010matrix, koltchinskii2011nuclear,foygel2011learning,klopp2014noisy,negahban2012restricted}.

The use of a spectral scaled Student prior is motivated from recent works in which it had been shown to lead to optimal rates in different problems such as high-dimensional regression \cite{dalalyan2012sparse}, image denoising \cite{dalalyan2020exponential} and reduced rank regression \cite{mai2022optimal}. Although it is not a conjugate prior in our problems, it allows to implement conveniently gradient-based sampling methods efficiently. For the computation of the proposed estimators and sampling from the quasi posterior, we employ a Langevin Monte Carlo (LMC) method. 

The rest of the paper is structured as follows. In Section \ref{sc_blr}, we introduce and study the problem of bilinear regression, the low-rank promoting prior distribution is also presented in this section. Section \ref{sc_imc} is devoted for the analysis of the problem of inductive matrix completion. In Section \ref{sc_numstudy}, we present shortly LMC methods for these two problems, after that some simulations are demonstrated. We conclude our works in Section \ref{sc_conclus}. All of the technical proofs are given in the appendix \ref{sc_appendix_proof}.

\subsection{Notations}
Let $\mathbb{R}^{n_1\times n_2}$ denote the set of $n_1 \times n_2$ matrices with real coefficients. Let $A^\intercal \in \mathbb{R}^{n_2\times n_1}$ denote the transpose of $A$.  For any $A\in \mathbb{R}^{n_1\times n_2}$ and $I=(i,j) \in\{1,\dots,n_1\} \times \{1,\dots,n_2\}$, we denote by $A_I=A_{(i,j)}=A_{i,j}$ the coefficient on the $i$-th row and $j$-th column of $A$. The matrix in $\mathbb{R}^{n_1\times n_2}$ with all entries equal to $0$ is denoted by $\mathbf{0}_{n_1 \times n_2}$. For a matrix $B\in\mathbb{R}^{n_1 \times n_1}$ we let ${\rm Tr}(B)$ denote its trace. The identity matrix in $ \mathbb{R}^{n_1 \times n_1} $ is denoted by $\mathbf{I}_{n_1}$. For $A\in \mathbb{R}^{n_1\times n_2}$, we define its sup-norm $\| A \|_{\infty} = \max_{i,j}|A_{i,j}| $; its Frobenius norm $ \|A\|_F$ is defined by $ \|A\|_F^2 = {\rm Tr}( A^\intercal A) = \sum_{i,j} A_{i,j}^2$ and ${\rm rank}(A)$ its rank.

\section{Bilinear linear regression}
\label{sc_blr}
\subsection{Model}
Let $ Y \in \mathbb{R}^{n\times q} $ consist of $ n $ independent response vectors, $ X \in \mathbb{R}^{n\times p} $ be a given between-individuals design matrix and $ Z \in \mathbb{R}^{k\times q} $ is a known  within-individuals design
matrix. Consider the bilinear regression model as follows
\begin{equation}
\label{model_blr}
Y = XM^*Z + E, 
\end{equation}
where $ M^* \in \mathbb{R}^{p\times k} $ is the unknown parameter matrix. The random noise matrix $ E $ is assumed to have zero mean,  $ \mathbb{E} (E) = 0. $ The main assumption here is the reduced-rank restriction on the model parameter that $ {\rm rank}(M^*) < \min (p,k) $.

In model \eqref{model_blr}, as soon as $ k=q $ and $ Z= \mathbf{I}_{q} $ we recover, as a special case, the reduced rank regression problem \cite{anderson1951estimating,izenman2008modern}. However, we consider in this paper the case where $ Z $ contains explanatory variables. The low-rank assumption can be interpreted as there exists a latent process so that the response data is not only affected by the “between-individuals” structure of the model but can also be interpretable in terms of the “within-individuals” structure. This model is often called the growth curve model or generalized multivariate analysis model
(GMANOVA), see \cite{von2018bilinear} for more details.

\begin{assume}
	\label{assum_bounded}
	There is a known constant $ C  < +\infty $  such that
	$
	\| X M^*Z \|_{\infty} \leq C.
	$
\end{assume}

\noindent From Assumption \ref{assum_bounded}, it is not reliable to return predictions $X MZ $ with entries that are outside of interval $[-C,C]$. However, for computational reason, it is extremely convenient to employ an unbounded prior for $M$. Therefore, we propose to use unbounded distributions for $M$, but to use, as a predictor, a truncated version of $XMZ $ rather than $ M $ itself. For a matrix $A$, let 
$
\Pi_C(A) = \arg\min_{\|B\|_{\infty} \leq C } \|A-B\|_F
$
be the orthogonal projection of $A$ on matrices with entries bounded by $C$. Note that $B$ is simply obtained by replacing entries of $A$ larger than $C$ by $C$, and entries smaller than $-C$ by $-C$.

For a matrix $ M\in \mathbb{R}^{p\times k} $, we denote by $r(M)$ the empirical risk of $M$ as
\begin{equation*}
r(M) 
= 
\dfrac{1}{nq}  \| Y - \Pi_C(X M Z) \|_F^2
\end{equation*}
and its expectation is denoted by
\begin{equation*}
R(M) = 
\mathbb{E} \left[ r(M)  \right]
= 
\mathbb{E} \left[  \left( Y_{11} - (\Pi_C(XMZ))_{_{11}} \right)^2  \right] .
\end{equation*}
The focus of our work in this paper is on the predictive aspects of the model that is a matrix $M$ predicts almost as well as $M^*$ if $ R(M) - R(M^*) $ is small. Under the assumption that $E_{ij} $ has a finite variance, using Pythagorean theorem, we have
\begin{equation}
\label{pyta}
R(M) - R(M^*) 
=  
\dfrac{1}{nq} \| \Pi_C(XMZ) - XM^* Z \|^2_{F}
\end{equation}
for any $ M $, which means that our results can also be interpreted in terms of Frobenius norm.

Let $\pi$ be a prior distribution on $\mathbb{R}^{p\times k}$ (see Subsection \ref{sc_priors}).
For any $\lambda>0 $, we define the quasi-posterior
\begin{equation*}
\widehat{\rho}_{\lambda} (dM) \propto \exp(-\lambda
r(M)) \pi (dM).
\end{equation*}
Note that, for $\lambda=nq/(2\sigma^2)$, this is exactly the posterior that we would be obtain for a Gaussian noise $ E_{ij} \sim \mathcal{N}(0,\sigma^2)$. However, our theoretical results will hold under a more general class of noise. It is known that a small enough $\lambda$ serves the purpose when the model is misspecified \cite{grunwald2017inconsistency}. Moreover, even in the case of a Gaussian noise, in high-dimensional settings, taking $\lambda$ smaller than $\lambda=n/(2\sigma^2)$ leads to better adaptation properties \cite{dalalyan2008aggregation,dalalyan2012sparse}. The choice of $\lambda$ in our method will be specified below.

We consider the following posterior mean of $XMZ $, given by
\begin{equation}
\label{equat_estimator}
\widehat{XMZ}_{\lambda} 
= 
\int \Pi_C(XMZ) \widehat{\rho}_{\lambda}(d M).
\end{equation}

\noindent Let us briefly discuss the above estimator. It is noted from our simulations experiments that using reasonable values for $C$, the Monte Carlo algorithm never samples matrices $M$ such that $\Pi_C(X M Z) \neq XM Z$. In other words, it has very little impact in practice and it is necessary for technical proofs. If one would like to obtain an estimator of $M^*$ instead of an estimator of $XM^*Z $, when $X^\intercal X$ and $ Z Z^\intercal $ are invertible, one can  consider 
	$
	\widehat{M}_\lambda 
	=
	(X^\intercal X)^{-1} X^\intercal \widehat{XMZ}_{\lambda} Z^\intercal (Z Z^\intercal)^{-1} $ and note that $X \widehat{M}_\lambda Z =  \widehat{XMZ}_{\lambda}$.

\noindent The quasi-posterior is often referred to as the so-called ``Gibbs posterior'' in the PAC-Bayes approach \cite{catoni2004statistical,catonibook,alquier2015properties,gue2019,TUTO}. Moreover, $\widehat{M}_{\lambda}$ is sometimes referred to as the Gibbs estimator, or the exponentially weighted aggregate (EWA) \cite{dalalyan2008aggregation,rigollet2012sparse,dalalyan2018exponentially}.

\subsection{Prior specification}
\label{sc_priors}
We consider, in this paper, the following spectral scaled Student prior, with parameter $\tau>0$,
\begin{align}  
\label{prior_scaled_Student}
\pi(M) 
\propto 
\det (\tau^2 \mathbf{I}_{m} + MM^\intercal )^{-(p+m+2)/2}.
\end{align}

One can verify that 
$
\pi(M) 
\propto 
\prod_{j=1}^{m}  (\tau^2  + s_j(M)^2 )^{- (p+m+2)/2 },
$
where $ s_j(M) $ denotes the $j^{th}$ largest singular value of $ M $. Intuitively,  one can recognize a scaled Student distribution evaluated at $ s_j(M)  $ in the last display above which induces approximate sparsity on the $s_j(M)$ \cite{dalalyan2012sparse}. Also, the log-sum function $ \sum_{j=1}^m \log (\tau ^2  + s_j(M)^2) $ used by \cite{candes2008enhancing,yang2018fast} is well known to enforce approximate sparsity on the singular values $s_j(M)$.  In other words, under this prior, most of the $s_j(M)$ are close to $0$, which means that $M$ is well approximated by a low-rank matrix. Thus it has the ability to promote the low-rankness of $ M $.

The spectral scaled Student prior has been used before in the context of image denoising \cite{dalalyan2020exponential} and in reduced rank regression \cite{mai2022optimal}. It can also be seen as the marginal distribution of the Gaussian-inverse Wishart prior where the precision matrix is integrated out that is explored in \cite{yang2018fast} for matrix completion. As mentioned in the introduction, this prior is not conjugate in our problems, however, it is particularly convenient to implement gradient-based sampling algorithms, the Langevin Monte Carlo, see Section \ref{sc_numstudy}.

\subsection{Theoretical results}

We assume the sub-exponential distribution assumption on the noise.

\begin{assume}
	\label{assum_noise1_blr}
	The entries $ E_{i,j}$ of $E $ are independent. There exist two known constants $\sigma>0$ and $\xi>0$ such that
	$ \forall k\geq 2,\, \mathbb{E} (|E_{i,j} |^{k}) \leq \sigma^{2} k! \xi^{k-2}/2.$
\end{assume}

\noindent Let us put 
$
C_1 = 8(\sigma^2 +C^2) ;\, C_2 = 64C \max(\xi,C);\,
\tau^* = \sqrt{C_1 (k+p)/(n k q \|X \|_{F}^2 \|Z \|_{F}^2 )}. $ The statistical properties of mean estimator are given in the following theorem where we propose a non-asymptotic analysis for our mean estimator.

\begin{theorem}
	\label{thrm_blr} 
	Let Assumptions \ref{assum_bounded} and \ref{assum_noise1_blr} be satisfied. Fix the parameter $\tau= \tau^* $ in the prior. Fix $\delta>0$ and define $\lambda^* := nq \min(1/(2C_2), \delta/[C_1(1+\delta)] ) $. Then, for any $\varepsilon\in(0,1)$, we have, with probability at least $1-\varepsilon$ on the sample,
	\begin{multline*}
	\left\|\widehat{XMZ}_{\lambda^*} -XM^*Z \right\|_{F}^2
	\leq 
	\inf_{0\leq r \leq pk} \inf_{
		\begin{tiny}
		\begin{array}{c}
		\bar{M}\in\mathbb{R}^{p\times k}
		\\
		{\rm rank}(\bar{M}) \leq r
		\end{array}
		\end{tiny}
	} \Biggl\{ (1+\delta) \left\|X\bar{M}Z - X M^*Z \right\|_{F}^2 \,
	+
	\\ 
	\frac{C_1 (1+\delta)^2}{\delta} 
	\left[  4 r (k+p+2) \log 
	\left( 1+\frac{\| X \|_F \|Z\|_F \| \bar{ M } \|_F}{ \sqrt{C_1}} \sqrt{ \frac{nkq}{r(k+p)}} \right) + k+p + 2 \log \frac{2}{\varepsilon} \right]
	\Biggr\}.
	\end{multline*}
\end{theorem}

The choice of $ \lambda = \lambda^* $ arises from the optimization of a upper bound on the risk $ R $ (in the proof of this theorem). However this choice may not be necessarily the best choice in practice even though it gives the good order of magnitude for $ \lambda $. The user could use cross-validation to properly tune the temperature parameter. It is also noted that ${\rm rank} (\bar{M})\neq 0$ is not required in the above formula, if ${\rm rank} (\bar{M})=0$ then $\bar{M}=0$ and we interpret $0\log(1+0/0)$ as $0$. The proof of this theorem relies on the PAC-Bayes theory and it is provided in the appendix \ref{sc_appendix_proof}.
In particular, we can upper bound the infimum on $\bar{M}$ by taking $\bar{M} = M^*$, which leads to the following result.

\begin{corollary}
	Under the same assumptions and the same $\tau,\lambda^*$ as in Theorem~\ref{thrm_blr}, let $r^*={\rm rank}(M^*)$. Then, for any $\varepsilon\in(0,1)$, we have, with probability at least $1-\varepsilon$ on the sample,
	\begin{multline*}
\left\|\widehat{XMZ}_{\lambda^*} -XM^*Z \right\|_{F}^2
\leq 
\frac{C_1 (1+\delta)^2}{\delta} 
\left[  4 r^* (k+p+2) \log 
\left( 1+\frac{\| X \|_F \|Z\|_F \| M^*\|_F}{ \sqrt{C_1}} \sqrt{ \frac{nkq}{r(k+p)}} \right) + \right.
\\
\left.
 k+p + 2 \log \frac{2}{\varepsilon} \right]
.
\end{multline*}
\end{corollary}

\noindent While Theorem \ref{thrm_blr} is on the statistical properties of the posterior mean, we further show the contraction of the posterior in the following theorem.

\begin{theorem}
	\label{thrm_contraction_blr} 
	Under the assumptions for Theorem~\ref{thrm_blr}, let $\varepsilon_n$ be any sequence in $(0,1)$ such that $\varepsilon_n\rightarrow 0$ when $n\rightarrow\infty$. Define
	\begin{multline*}
	\mathcal{M}_n 
	= 
	\Biggl\{ M \in \mathbb{R}^{p\times k}: 	\left\|\widehat{XMZ}_{\lambda^*} -XM^*Z \right\|_{F}^2
	\leq 
	\inf_{0\leq r \leq pk} \inf_{
		\begin{tiny}
		\begin{array}{c}
		\bar{M}\in\mathbb{R}^{p\times k}
		\\
		{\rm rank}(\bar{M}) \leq r
		\end{array}
		\end{tiny}
	} \Biggl\{ (1+\delta) \left\|X\bar{M}Z - X M^*Z \right\|_{F}^2 \,
	+
	\\ 
	\frac{C_1 (1+\delta)^2}{\delta} 
	\left[  4 r (k+p+2) \log 
	\left( 1+\frac{\| X \|_F \|Z\|_F \| \bar{ M } \|_F}{ \sqrt{C_1}} \sqrt{ \frac{nkq}{r(k+p)}} \right) + k+p + 2 \log \frac{2}{\varepsilon_n} \right]
	\Biggr\}.
	\end{multline*}
	Then
	$$ \mathbb{E} \Bigl[ \mathbb{P}_{M\sim \widehat{\rho}_{\lambda}} (M\in\mathcal{M}_n) \Bigr] \geq 1-\varepsilon_n \xrightarrow[n\rightarrow \infty]{} 1.  $$
\end{theorem}
The proof of this theorem is provided in the appendix \ref{sc_appendix_proof}.

\section{Inductive matrix completion}
\label{sc_imc}
\subsection{Model and method}
In the context of inductive matrix completion, given two side information matrices $ X$ and $Z $, we assume that only a random subset $ \Omega  $ of the entries of $ Y $ in model \eqref{model_blr} is observed. More precisely, we assume that we observe $m$ i.i.d random pairs $ (\mathcal{I}_1, Y_1),
\ldots, (\mathcal{I}_m, Y_m) $ given by
\begin{equation}
\label{main_model}
Y_i = (XM^*Z)_{\mathcal{I}_i} + \mathcal{E}_i, \quad i = 1, \ldots, m
\end{equation}
where $ M^* \in \mathbb{R}^{p\times k} $ is the unknown parameter matrix expected to be low-rank. The noise variables $ \mathcal{E}_i $ are assumed to be independent with
$ \mathbb{E} (\mathcal{E}_i) = 0. $ The variables $ \mathcal{I}_i $ are i.i.d copies of a random variable $ \mathcal{I} $ having distribution $ \Pi $ on the set
$ \lbrace 1, \ldots, n \rbrace \times \lbrace 1, \ldots, q \rbrace $, we denote $ \Pi_{x,y} := \Pi(\mathcal{I} = (x,y)) $.

When $ p=n, k=q $ and that $X=\mathbf{I}_n, Z= \mathbf{I}_{q} $ are the identity matrices, we recover the matrix completion as a special case~\cite{koltchinskii2011nuclear}; whereas when $ k=q $ and that $Z= \mathbf{I}_{q} $ is the identity matrix, we recover the reduced rank regression problem with incomplete response as a special case \cite{luo2018leveraging,mai2022optimal}. However, in the context of inductive matrix completion, we focus in the cases where $X$ and $ Z $ contains explanatory variables. In order words, the side information from $ n $ users and the side information from the $ q $ items will be taken into account in this model \cite{zilber2022inductive}.

It is also noted that there are two approaches can be used to model the observed entries of $Y$: with, or without replacement. For matrix completion, both had been studied, see for example \cite{candes2010matrix} for the case without replacement and \cite{koltchinskii2011nuclear} with replacement. Both settings have their own practical applications, and the same estimation methods are used in both cases. In this paper, we study the case for i.i.d variables $ \mathcal{I}_i $, which means that it is possible to observe the same entry multiple times. Note that, from the results in Section 6 of~\cite{hoeffding1963probability}, our results can be extended directly to the case of sampling $\mathcal{I}_i$ without replacement on the condition that they are sampled uniformly, and that there is no observation noise: $\mathcal{E}_i=0$.

We are now adapting the quasi-Bayesian approach for bilinear regression in Section \ref{sc_blr} to the context of inductive matrix completion. For a probability distribution $P$ on $\{1,\dots,n_1\} \times\{1,\dots,n_2\}$, we generalize the Frobenius norm by $ \|A\|_{F,P}^2 = \sum_{i,j} P[(i,j)] A_{i,j}^2 $; note that when $P$ is the uniform distribution, then $ \|A\|_{F,P}^2 = \|A\|_{F}^2 /(n_1 n_2) $. 

For a matrix $ M\in \mathbb{R}^{p\times k} $, we denote the empirical risk of $M$, $ r'(M) $, and its expected risk $ R'(M)  $ respectively as
$$
r'(M) 
= 
\frac{1}{m}
\sum_{i = 1}^{m} \left( Y_i - (\Pi_C(X MZ))_{\mathcal{I}_i} \right) ^2,
\quad
R'(M) = \mathbb{E} [ r'(M)  ]
= 
\mathbb{E} [  \left( Y_1 - (\Pi_C(XMZ))_{\mathcal{I}_1} \right)^2  ] .
$$
As in Section \ref{sc_blr}, we will focus on the predictive aspects of the model that is a matrix $M$ predicts almost as well as $M^*$ if $ R'(M) - R'(M^*) $ is small. Under the assumption that $\mathcal{E}_i$ has a finite variance, thanks to the Pythagorean theorem, we have
\begin{equation}
\label{pyta_imc}
R'(M) - R'(M^*) =  \| \Pi_C(XMZ) - XM^*Z  \|^2_{F,\Pi}
\end{equation}
for any $ M $, which means that our results can also be interpreted in terms of estimation of $M^*$ with respect to a generalized Frobenius norm.

Here, the prior $\pi$ is the low-rank inducing prior specified in the Subsection \ref{sc_priors} above.
For any $\lambda>0$, we define the quasi-posterior
\begin{equation*}
\hat{\rho}'_{\lambda} (dM) 
\propto 
\exp(-\lambda r'(M)) \pi (dM).
\end{equation*}
We will actually specify our choice of $\lambda$ below.

The truncated posterior mean of $XMZ $ is given by
\begin{equation}
\label{equat_estimator_imc}
\widehat{XMZ }_{\lambda} 
= 
\int \Pi_C(XM) \hat{\rho}'_{\lambda}(d M).
\end{equation}

\noindent Here, for the same technical reasons as in the context of bilinear regression problem, this truncation has a very little impact in practice for reasonable values of $ C $.

\subsection{Theoretical results}

In this section, we derive the statistical properties of the posterior $\hat{\rho}'_\lambda$ and the mean estimator $\widehat{XMZ}_\lambda$ for the context of inductive matrix completion. Let us first state our assumptions on this model.
\begin{assume}
	\label{assum_noise}
	The noise variables  $ \mathcal{E}_1, \ldots, \mathcal{E}_m $ are independent and independent of $\mathcal{I}_1,\ldots,\mathcal{I}_{m}$. There exist two known constants $\sigma' >0 $ and $\xi' >0 $ such that
	$ \forall k\geq 2,\, \mathbb{E} (|\mathcal{E}_{i}|^{k}) \leq \sigma'^{2} k! \xi'^{k-2}/2.$
\end{assume}

\noindent Assumption \ref{assum_bounded} and \ref{assum_noise} are both standard, they have been used in \cite{luo2018leveraging} for theoretical analysis of reduced rank regression and in \cite{koltchinskii2011nuclear} for trace regression and matrix completion.

Let us put 
$
C'_1 = 8(\sigma'^2 +C^2) ;\, C'_2 = 64C \max(\xi',C);\,
\tau^* = \sqrt{C'_1 (k+p)/(m k p \|X \|_{F}^2 \|Z \|_{F}^2 )}. 
$

\begin{theorem}
\label{thrm_main} 
Let Assumptions \ref{assum_bounded} and \ref{assum_noise} be satisfied. Fix the parameter $\tau= \tau^* $ in the prior. Fix $\delta>0$ and define $\lambda'^* := m \min(1/(2C'_2), \delta/[C'_1(1+\delta)] ) $. Then, for any $\varepsilon\in(0,1)$, we have, with probability at least $1-\varepsilon$ on the sample,
\begin{multline*}
 \left\|\widehat{XMZ}_{\lambda'^*} -XM^*Z \right\|_{F,\Pi}^2
\leq 
\inf_{0\leq r \leq pk} \inf_{
\begin{tiny}
\begin{array}{c}
\bar{M}\in\mathbb{R}^{p\times k}
\\
{\rm rank}(\bar{M}) \leq r
\end{array}
\end{tiny}
} \Biggl\{ (1+\delta) \|X\bar{M}Z - X M^*Z \|_{F,\Pi}^2
+
\\ 
 \frac{C'_1 (1+\delta)^2}{\delta} \frac{ \left( 4 r (k+p+2) \log \left( 1+\frac{\| X \|_F \|Z \|_F \| \bar{ M } \|_F}{ \sqrt{C_1}} 
 	\sqrt{ \frac{mkp}{r(k+p)}} \right) 
 	+
 	k+p + 2 \log \frac{2}{\varepsilon} \right)}{ m} \Biggr\}.
\end{multline*}
\end{theorem}

Similar to the context bilinear regression, the choices of $ \lambda = \lambda^*, \tau= \tau^* $ come from the optimization of a upper bound on the risk $ R $ (in the proof of this theorem). Therefore these choices may not be necessarily the best choice in practice even though it gives the good order of magnitude for tuning these parameters. The user could use cross-validation to properly tune them in practice.
Note again that ${\rm rank} (\bar{M})\neq 0$ is not required in the above formula, if ${\rm rank} (\bar{M})=0$ then $\bar{M}=0$ and we interpret $0\log(1+0/0)$ as $0$. The proof of this theorem is provided in the appendix \ref{sc_appendix_proof}.
In particular, we can upper bound the infimum on $\bar{M}$ by taking $\bar{M} = M^*$, which leads to the following result.

\begin{corollary}
Under the assumptions that Theorem~\ref{thrm_main} holds, let $r^*={\rm rank}(M^*)$. Put
$$
R_{\delta,m,p,k,r^*,\varepsilon}
 : =   
 \frac{C'_1 (1+\delta)^2}{\delta} \frac{ \left( 4 r (k+p+2) \log \left( 1+\frac{\| X \|_F \|Z \|_F \| \bar{ M } \|_F}{ \sqrt{C_1}} 
	\sqrt{ \frac{mkp}{r(k+p)}} \right) +k+p + 2 \log \frac{2}{\varepsilon} \right)}{ m},
$$
 then 
\begin{align*}
 \left\|\widehat{XMZ}_{\lambda'^*} -XM^*Z \right\|_{F,\Pi}^2
\leq 
R_{\delta,m,p,k,r^*,\varepsilon} 
\end{align*}
and in particular, if the sampling distribution $\Pi$ is uniform,
\begin{align*}
 \dfrac{ \| \widehat{XMZ}_{\lambda'^*} - XM^*Z \|^2_{F} }{ nq}
\leq 
R_{\delta,m,p,k,r^*,\varepsilon} .
\end{align*}
\end{corollary}

\begin{remark}
Up to a log-term, our error rate $ r(k+p)/m $ is similar to the best known up-to-date rate derived in \cite{zilber2022inductive}.
\end{remark}

While Theorem~\ref{thrm_main} is about the finite sample convergence rate of the posterior mean, it is actually possible to prove that the quasi-posterior $\hat{\rho}'_\lambda$ contracts around $M^*$ at the same rate.
\begin{theorem}
\label{thrm_contraction} 
Under the same assumptions for Theorem~\ref{thrm_main}, and the same definition for $\tau$ and $\lambda^*$, let $\varepsilon_m $ be any sequence in $(0,1)$ such that $\varepsilon_m \rightarrow 0$ when $ m \rightarrow\infty $. Define
\begin{multline*}
\mathcal{E}_m 
 = 
 \Biggl\{ M \in \mathbb{R}^{p\times k}: \| \Pi_C(XMZ) - XM^*Z  \|^2_{F,\Pi} 
\leq
\inf_{1\leq r \leq pk} \inf_{
\begin{tiny}
\begin{array}{c}
\bar{M}\in\mathbb{R}^{p\times k}
\\
{\rm rank}(\bar{M}) \leq r
\end{array}
\end{tiny}
} \Biggl[ (1+\delta) \|X\bar{M}Z - X M^*Z \|_{F,\Pi}^2 +
\\ 
\frac{C'_1 (1+\delta)^2}{\delta} \frac{ \left( 4 r (k+p+2) \log \left( 1+\frac{\| X \|_F \|Z \|_F \| \bar{ M } \|_F}{ \sqrt{C_1}}
	 \sqrt{ \frac{mkp}{r(k+p)}} \right) +k+p + 2 \log \frac{2}{\varepsilon_m} \right)}{ m} \Biggr] \Biggr\}.
\end{multline*}
Then
$$ \mathbb{E} \Bigl[ \mathbb{P}_{M\sim \hat{\rho}'_{\lambda}} (M\in\mathcal{E}_m) \Bigr] 
\geq 
1-\varepsilon_m \xrightarrow[m \rightarrow \infty]{} 1.  $$
\end{theorem}
The proof of this theorem is provided in the appendix \ref{sc_appendix_proof}.

\section{Numerical studies}
\label{sc_numstudy}

\subsection{Langevin Monte Carlo implementation}
In this section, we propose to sample from the (quasi) posterior, in Section \ref{sc_blr} and \ref{sc_imc}, by a suitable version of the Langevin Monte Carlo (LMC) algorithm, a gradient-based sampling method. We propose use a constant step-size unadjusted LMC algorithm, see \cite{durmus2019high} for more details. The algorithm is given by an initial matrix $M_0$ and the recursion
\begin{equation}
\label{langevinMC}
M_{k+1}  =  M_{k} - h\nabla \log    \widehat{\rho}_{\lambda}(M_k) +\sqrt{2h}\,N_k\qquad k=0,1,\ldots
\end{equation}
where $h>0$ is the step-size and $ N_0, N_1,\ldots$ are independent random matrices with i.i.d standard Gaussian entries. We provide a pseudo-code for LMC in Algorithm \ref{lmc_algorithm}. For small values of the step-size $ h $, the output of Algorithm \ref{lmc_algorithm}, $ \widehat{M} $,   is very close to the integral \eqref{equat_estimator} of interest. However, for some $h$ that may not small enough, the sum can explode \cite{roberts2002langevin}. In such case, we consider to include a Metropolis–Hastings correction in the algorithm. Another possible choice is to take a smaller $ h $ and restart the algorithm, although it slows down the algorithm but we keep some control on its time of execution. On the other hand, the Metropolis–Hastings approach ensures the convergence to the desired distribution, however, the algorithm is greatly slowed down because of an additional acception/rejection step at each iteration.

Next, we propose a Metropolis-Hasting correction to the LMC algorithm. It guarantees the convergence to the (quasi) posterior and it also provides a useful way for choosing $ h $. More precisely, we consider the update rule in \eqref{langevinMC} as a proposal for a new candidate,
\begin{align}
\tilde{M}_{k+1} = M_{k} - h\nabla \log    \widehat{\rho}_{\lambda}  (M_k) +\sqrt{2h}\,N_k,\qquad
k=0,1,\ldots,
\label{mala}
\end{align}
Note that the matrix $\tilde{M}_{k+1} $ is normally distributed with mean $ M_{k} - h\nabla \log    \widehat{\rho}_{\lambda}  (M_k) $ and the covariance matrices equal to $ 2h $ times the identity matrices. This proposal is then accepted or rejected according to the Metropolis-Hastings algorithm that the proposal is accepted with probability:
\begin{equation}
\label{eq_ac_mala}
A_{MALA} 
:=
\min \left\lbrace 1,  \frac{  \widehat{\rho}_{\lambda}  (\tilde{M}_{k+1}) q(M_k | \tilde{M}_{k+1}) }{
	\widehat{\rho}_{\lambda}  (M_k ) q(\tilde{M}_{k+1} | M_k ) } \right\rbrace,
\end{equation}
where 
$$
q(x' | x) \propto \exp \left(-\frac{1}{4h}\|x'-x + h\nabla \log    \widehat{\rho}_{\lambda}  (x) \|^2_F \right)
$$
is the transition probability density from $x$ to $x'$. The details of the Metropolis-adjusted Langevin algorithm (denoted by MALA) are presented Algorithm \ref{mala_algoritm}. Compared to random-walk Metropolis–Hastings, MALA usually proposes moves into regions of higher probability, which are then more likely to be accepted.

We note that the step-size $h$ for MALA is chosen such that the acceptance rate is approximate $0.5$ following \cite{roberts1998optimal}, while the step-size for LMC in the same setting should be smaller than the one for MALA \cite{dalalyan2017theoretical}.

\subsection{Simulation studies for biliear regression}
\label{sc_simu_blr}

We perform some numerical studies on simulated data to access the performance of our proposed algorithms. All
simulations were conducted using the \texttt{R} statistical software \cite{rsoftware}.

For fixed dimensions $ q=10, k =20 $ of the data, we vary $ n =100 $ and $ n=1000 $ to check the effect of the samples whereas the dimensions of the coefficient matrix are varied by $ p=10 $ and $ p=100 $. The entries of the design matrices $X $ and $ Z $ are independently simulated from the standard Gaussian $ \mathcal{N}(0, 1) $. Then, given a matrix $ M^* $, we simulate the response matrix $ Y $ from model \eqref{model_blr} whose entries of the noise matrix $ E $ are i.i.d sampled from $ \mathcal{N}(0, 1) $.  We consider the following setups for the true coefficient matrix:
\begin{itemize}
	\item Model I: The true coefficient matrix $ M^* $ is a rank-2 matrix that is generated as $ M^* = B_1B_2^\top$ where $B_1 \in \mathbb{R}^{p\times 2}, B_2 \in \mathbb{R}^{k\times 2} $ and all entries in $ B_1 $ and $B_2$ are i.i.d sampled from $ \mathcal{N}(0, 1) $.
	
	\item Model II: An approximate low-rank set up is studied. This series of simulation is similar to the Model I, except that the true coefficient is no longer rank 2, but it can be well approximated by a rank 2 matrix:
	$$
	M^* = 2\cdot B_1B_2^\top + U,
	$$
	where $U $ is a matrix whose entries are i.i.d sampled from $ \mathcal{N}(0, 0.1) $.
\end{itemize}

We compare our approaches denoted by LMC and MALA against the (generalized) ordinary least square \cite{von2018bilinear}, denoted by OLS. The OLS is defined as follow
$$
\hat{M}_{\rm OLS} = (X^\top X)^\dagger X^\top Y Z^\top (ZZ^\top )^\dagger
$$ 
where $ A^\dagger $ denotes the Moore-Penrose matrix pseudo-inversion. We fixed $ \lambda = nq, \tau =1 $ and the LMC and MALA methods are initiated at the OLS estimator and are run with 10000 iterations where the first 1000 steps are removed as burn-in periods.

The evaluations are done by using the mean squared estimation error (Est) and the normalized (relative) mean square error (Nmse)
$$
{\rm Est} := \| \hat{M} -M^* \|^2_F/ (pk),
\quad
{\rm Nmse}:= \| \hat{M} -M^*\|^2_F/ \| M^* \|^2_F ,
$$
and the prediction error (Pred) as 
$$ 
{\rm Pred} := \| X(\hat{M} -M^*) Z \|^2_F/ (nq),
$$ 
where $ \hat{M} $ here is one of the estimator from LMC, MALA or OLS. We report the averages and the standard deviation of these errors over 100 data replications.

The results are given in Table \ref{tb_model1} and Table \ref{tb_model2}. We can observe that our proposed methods behave quite similar to OLS method. However, the estimation method obtained from MALA algorithm often yields smaller prediction errors, especially in high-dimensional settings. This feature is even more visible in the context of inductive matrix completion in the next subsection.

\subsection{Simulation studies for inductive matrix completion}

The simulation settings for inductive matrix completion are similar to the settings for bilinear regression, subsection \ref{sc_simu_blr}. However, after obtaining the response matrix $ Y $, we remove uniformly at random $ \kappa = 10\% $ and $ \kappa = 30\% $ of the entries of $ Y $. Here $ \kappa $ denotes the missing rate. We denote the response matrix with missing entries by $ Y_{\rm miss} $.

As in the context of inductive matrix completion, we only observe the response matrix with missing entries, $ Y_{\rm miss} $, and thus we can not construct the OLS estimator as in the case of bilinear regression. For this purpose, we first impute the missing entries in $ Y_{\rm miss} $ by using the \texttt{R} package \texttt{softImpute} \cite{softimputepackage} where the rank of $ M^* $ is specified as the true rank for matrix $ Y_{\rm miss} $. We denote the resulting imputed matrix by $ Y_{\rm imp} $.

We compare our approaches denoted by LMC and MALA against the (imputed and generalized) ordinary least square, denoted by OLS\_imp. The OLS\_imp is defined as follow
$$
\hat{M}_{\rm OLS\_imp} = (X^\top X)^\dagger X^\top Y_{\rm imp} Z^\top (ZZ^\top )^\dagger
$$ 
where $ A^\dagger $ denotes the Moore-Penrose matrix pseudo-inversion. The LMC and MALA methods are initiated at the OLS\_imp estimator and are run with 10000 iterations where the first 1000 steps are removed as burn-in periods.

As in subsection \ref{sc_simu_blr}, we report the averages and the standard deviation of the the mean squared estimation error (Est), the normalized (relative) mean square error (Nmse) and of the prediction error (Pred) over 100 data replications.

The results are given in Table \ref{tb_imc_model1} and Table \ref{tb_imc_model2}. Here, it is more clear that the results from our MALA method outperforms the other methods in term of prediction error in most considered settings. It is more significant when the missing rate in the response matrix increases. It is also noted that our MALA method robustly works well in the case of approximate low-rank setting (model II) while the OLS and LMC are not.

\section{Discussion and conclusion}
\label{sc_conclus}
In this work, we presented a unified quasi-Bayesian analysis for the problem of bilinear regression and for the problem of inductive matrix completion. We focused on filling the gaps in theoretical understanding of Bayesian-like methods in these problem. By leveraging PAC-Bayesian bound techniques, we proved some non-asymptotic results for our proposed estimator, we also derived the contraction properties for the corresponding (quasi) posteriors. 

We further propose an efficient gradient-based sampling algorithm to sample from the (quasi) posterior and to approximately compute the mean estimators. The resulting methods, denoted by LMC and MALA, were shown to work well in some simulation studies and in compared with a ordinary least squared method.

Some questions remain open: such as missing entries in the covariate matrices $ X,Z $ or relaxing the i.i.d assumption by considering models with dispersion matrix. These might be the objects of future works

\subsection*{Acknowledgments}
TTM is supported by the Norwegian Research Council, grant number 309960 through the Centre for Geophysical Forecasting at NTNU. The R codes and data used in the numerical experiments are available at:  \url{https://github.com/tienmt/blr_imc} .

\subsection*{Conflicts of interest/Competing interests}
The authors declare no potential conflict of interests.

\clearpage

\begin{table}[!ht]
	\small
	\caption{\it \small Simulation results on simulated data in Model I in inductive matrix completion, with fixed $ q=10, k =20 $, for different methods, with their standard error in parentheses over 100 replications. ($\kappa $ is the missing rate; Est: average of estimation errors; Pred: average of prediction errors; Nmse: average of normalized estimation errors).}
	\small
	\centering
	\begin{tabular}{p{1.4cm} | p{1.3cm} | p{14mm}|cccc} 
		\toprule\toprule
		&& Errors    & LMC & MALA & OLS\_imp  & 
		\\ \midrule
		&& Est 	& 1.0559 (.5060) & 1.0803 (.5122)& 1.0559 (.5060)&
		\\ 
$ n=100 $  & $ p=10 $ 
		& Pred  & 0.1028 (.0193) & 0.1082 (.0143)& \textbf{0.1020} (.0197)&
		\\ 
$ \kappa= 10\% $		&
		& Nmse  & 0.4986 (.1116) & 0.5139 (.1197)& 0.4986 (.1116)&
		\\ 
		\cline{2-7}
		&
		& Est 	& 1.4008 (.8555) & \textbf{1.3987} (.8542) & 1.4009 (.8555)& 
		\\ 
&$ p=100 $ & 
		Pred  & \textbf{1.2250} (.4568) & 1.4468 (.4137) & 1.2252 (.4570)& 
		\\ 
		&& Nmse  & 0.7148 (.3591) & \textbf{0.7136} (.3581) & 0.7148 (.3591) & 
		\\ 
		\midrule 
		\hline
		&& Est 	& 1.0432 (.4963) & 1.0917 (.5085)& 1.0432 (.4963)&
		\\ 
$ n=100 $  &$ p=10 $ 
		& Pred  & 0.2402 (.2705) & \textbf{0.1447} (.0204)& 0.2446 (.2780)&
		\\ 
$\kappa= 30\%  $	&
		& Nmse  & 0.5242 (.1257) & 0.5538 (.1335)& 0.5242 (.1257)&
		\\ 
		\cline{2-7}
		&
		& Est & 1.6242 (.8179) & \textbf{1.6224} (.8169)& 1.6242 (.8179)&
		\\ 
&$ p=100 $ & 
		Pred  & \textbf{9.8879} (14.11) & 10.807 (13.84)& 9.8901 (14.11)&
		\\ 
		&& Nmse  & 0.7993 (.3340) & \textbf{0.7985} (.3334)& 0.7993 (.3340)&
		\\ 
		\midrule 
		\hline
		&& Est 	& 0.9810 (.4532) & 0.9882 (.4478)& 0.9810 (.4532)&
		\\ 
$ n=1000 $  &$ p=10 $ 
		& Pred  & 0.0114 (.0033) & \textbf{0.0112} (.0015)& 0.0114 (.0033)&
		\\ 
$\kappa= 10\%  $	&
		& Nmse  & 0.4933 (.1076) & 0.4984 (.1075)& 0.4933 (.1076)&
		\\ 
		\cline{2-7}
		&
		& Est & 1.0063 (.3465) & 1.0088 (.3471)& 1.0063 (.3465)&
		\\ 
&$ p=100 $ & 
		Pred  & 0.1902 (.1758) & \textbf{0.1116} (.0049)& 0.1902 (.1759)&
		\\ 
		&& Nmse  & 0.5069 (.1049) & 0.5082 (.1050)& 0.5069 (.1049)&
		\\ 
		\midrule 
		\hline
		&& Est & 1.0110 (.4886) & 1.0223 (.4872)& 1.0110 (.4886)&
		\\ 
$ n=1000 $ &$ p=10 $ & 
		Pred  & 0.0539 (.0599) & \textbf{0.0141} (.0019)& 0.0540 (.0599)&
		\\ 
$\kappa= 30\%  $	&
		& Nmse  & 0.5129 (.1030) & 0.5206 (.1043)& 0.5129 (.1030)&
		\\ 
		\cline{2-7}
		&& Est & 1.0291 (.3567) & 1.0312 (.3555)& 1.0291 (.3567)&
		\\ 
&$ p=100 $ & Pred  & 1.7529 (1.914) & \textbf{0.1475} (.0078)& 1.7530 (1.913)&
		\\ 
		&& Nmse  & 0.5054 (.1055) & 0.5067 (.1053)& 0.5054 (.1055)&
		\\ \bottomrule\bottomrule
	\end{tabular}
	\label{tb_imc_model1}
\end{table}

\begin{table}[!ht]
	\small
	\caption{\it \small Simulation results on simulated data in Model I in bilinear regression, with fixed $ q=10, k =20 $, for different methods, with their standard error in parentheses over 100 replications. (Est: average of estimation errors; Pred: average of prediction errors; Nmse: average of normalized estimation errors).}
	\small
	\centering
	\begin{tabular}{p{1.4cm} | p{1.3cm} | p{12mm}|cccc} 
		\toprule\toprule
		&& Errors    & LMC & MALA & OLS  & 
		\\ \midrule
		&& Est 	& 1.0053 (.5480) & 1.0342 (.5559) & \textbf{1.0052} (.5478) 
		\\ 
		$ n=100 $ &$ p=10 $ 
		& Pred  & 0.1138 (.0171) & \textbf{0.0985} (.0151) & 0.1014 (.0154) 
		\\ 
		&& Nmse  & 0.4931 (.1178) & 0.5100 (.1207) & \textbf{0.4930}(.1178) 
		\\ 
		\cline{2-7}
		&
		& Est 	& 1.3544 (.5867) & \textbf{1.3384} (.5836) & 1.3544 (.5867) 
		\\ 
		&$ p=100 $ & 
		Pred  & 1.0066 (.0430) & \textbf{0.8761} (.0756) & 1.0030 (.0424) 
		\\ 
		&& Nmse  & 0.7049 (.2944) & \textbf{0.6963} (.2927) & 0.7049 (.2944) 
		\\ 
		\midrule \hline
		&& Est 	& 1.0776 (.5671) & 1.0900 (.5670) & 1.0776 (.5671) 
		\\ 
		$ n=1000 $ &$ p=10 $ & 
		Pred  & 0.0099 (.0013) & 0.0099 (.0013) & 0.0099 (.0013) 
		\\ 
		&& Nmse  & 0.5185 (.1198) & 0.5264 (.1219) & 0.5185 (.1198) 
		\\ 
		\cline{2-7}
		&& Est 	& 0.9662 (.3240) & 0.9688 (.3244) & 0.9662 (.3240) 
		\\ 
		&$ p=100 $ & Pred  & 0.0999 (.0051) & \textbf{0.0989} (.0049) & 0.0998 (.0051) 
		\\ 
		&& Nmse  & 0.4961 (.1183) & 0.4976 (.1191) & 0.4961 (.1183) 
		\\ \bottomrule\bottomrule
	\end{tabular}
	\label{tb_model1}
\end{table}

\begin{table}[!ht]
	\small
	\caption{\it \small Simulation results on simulated data in Model II (approximate low-rank) in bilinear regression, with fixed $ q=10, k =20 $, for different methods, with their standard error in parentheses over 100 replications.  (Est: average of estimation errors; Pred: average of prediction errors; Nmse: average of normalized estimation errors).}
	\small
	\centering
	\begin{tabular}{p{1.4cm} | p{1.3cm} | p{14mm}|cccc} 
		\toprule\toprule
		&& Errors    & LMC & MALA & OLS  &
		\\ \midrule
		&& Est 	& 4.0731 (1.828) & 4.0989 (1.821) & 4.0731 (1.828) 
		\\ 
		$ n=100 $ &$ p=10 $ & Pred  & 0.1090 (.0160) & \textbf{0.0969} (.0140) & 0.0987 (.0145) 
		\\ 
		&& Nmse  & 0.5119 (.1226) & 0.5162 (.1241) & \textbf{0.5118} (.1226) 
		\\ 
		\cline{2-7}
		&
		& Est 	& 4.6047 (1.812) & \textbf{4.6038} (1.813) & 4.6047 (1.812) 
		\\ 
		&$ p=100 $ & Pred  & 1.0062 (.0462) & 1.0597 (.0495) & \textbf{1.0006} (.0469) 
		\\ 
		&& Nmse  & 0.5801 (.1942) & \textbf{0.5800} (.1941) & 0.5801 (.1942) 
		\\ 
		\midrule
		\hline
		&& Est 	& 3.6733 (1.606) & 3.6884 (1.606) & 3.6733 (1.606) 
		\\ 
		$ n=1000 $ &$ p=10 $ & Pred  & 0.0098 (.0015) & 0.0098 (.0015) & 0.0098 (.0015) 
		\\ 
		&& Nmse  & \textbf{0.4812} (.1271) & 0.4835 (.1260) & 0.4813 (.1271) 
		\\ 
		\cline{2-7}
		&& Est 	& 3.9972 (1.375) & 3.9986 (1.376) & 3.9972 (1.375) 
		\\ 
		&$ p=100 $ & Pred  & 0.1000 (.0043) & 0.1032 (.0057) & \textbf{0.0999} (.0043) 
		\\ 
		&& Nmse  & 0.5013 (.1061) & 0.5014 (.1063) & 0.5013 (.1062) 
		\\ \bottomrule\bottomrule
	\end{tabular}
	\label{tb_model2}
\end{table}

\begin{table}[!ht]
	\small
	\caption{\it \small Simulation results on simulated data in Model II (approximate low-rank) in inductive matrix completion, with fixed $ q=10, k =20 $, for different methods, with their standard error in parentheses over 100 replications. ($\kappa $ is the missing rate; Est: average of estimation errors; Pred: average of prediction errors; Nmse: average of normalized estimation errors).}
	\small
	\centering
	\begin{tabular}{p{1.4cm} | p{1.3cm} | p{14mm}|cccc} 
		\toprule\toprule
		&& Errors   & LMC & MALA & OLS\_imp  & 
		\\ \midrule
		&& Est 	& 3.8319 (1.691) & 3.8749 (1.719)& 3.8319 (1.690)&
		\\ 
$ n=100 $  & $ p=10 $ 
		& Pred  & 0.1604 (.1271) & \textbf{0.1092} (.0153)& 0.1598 (.1322)&
		\\ 
imis 10\%		&
		& Nmse  & 0.5116 (.1154) & 0.5169 (.1147)& 0.5116 (.1155)&
		\\ 
		\cline{2-7}
		&
		& Est & 5.9500 (2.834) & 5.9452 (2.835)& 5.9500 (2.834)&
		\\ 
&$ p=100 $ & 
		Pred  & 4.7640 (5.272) & \textbf{4.6964} (5.515)& 4.7658 (5.275)&
		\\ 
	&& Nmse  & 0.7313 (.3454) & 0.7307 (.3455) & 0.7313 (.3454)&
		\\ 
		\midrule 
		\hline 
		&& Est 	& \textbf{4.1838} (1.850) & 4.2535 (1.859)& 4.1839 (1.850)&
		\\ 
$ n=100 $  &$ p=10 $ 
		& Pred  & 0.7221 (.7562) & \textbf{0.1498} (.0183)& 0.7371 (.7741)&
		\\ 
imis 30\% 	&
		& Nmse  & 0.5182 (.1128) & 0.5283 (.1147)& 0.5182 (.1128)&
		\\ 
		\cline{2-7}
		&
		& Est & 7.1589 (4.084) & 7.1558 (4.083)& 7.1589 (4.084)&
		\\ 
&$ p=100 $ & 
		Pred  & \textbf{39.899} (52.40) & 40.233 (51.76)& 39.908 (52.41)&
		\\ 
	&& Nmse  & 0.8998 (.3821) & \textbf{0.8994} (.3820)& 0.8998 (.3821)&
		\\ 
		\midrule 
		\hline
		&& Est 	& 3.9618 (1.678) & 3.9788 (1.677)& 3.9618 (1.678)&
		\\ 
$ n=1000 $  &$ p=10 $ 
		& Pred  & 0.0409 (.0269) & \textbf{0.0110} (.0015)& 0.0409 (.0269)&
		\\ 
imis 10\% 	&
		& Nmse  & 0.4968 (.1196) & 0.4989 (.1195)& 0.4968 (.1196) &
		\\ 
		\cline{2-7}
		&
		& Est & 4.1153 (1.295) & 4.1163 (1.294)& 4.1153 (1.295)&
		\\ 
&$ p=100 $ & 
		Pred  & 1.0250 (.9988) & \textbf{0.1135} (.0051)& 1.0250 (.9988)&
		\\ 
	&& Nmse  & 0.5060 (.1096) & 0.5062 (.1096)& 0.5060 (.1096)&
		\\ 
		\midrule 
		\hline
		&& Est & 4.1647 (1.990) & 4.1836 (1.995) & 4.1647 (1.990) &
		\\ 
$ n=1000 $ &$ p=10 $ & 
		Pred  & 0.4615 (.3497) & \textbf{0.0141} (.0017) & 0.4616 (.3498) &
		\\ 
imis 30\%	&
		& Nmse  & 0.4905 (.1157) & 0.4933 (.1171) & 0.4905 (.1157) &
		\\ 
		\cline{2-7}
		&& Est 	& 4.0578 (1.400) & \textbf{4.0565} (1.397) & 4.0578 (1.400) &
		\\ 
&$ p=100 $ & Pred  & 8.5608 (6.419) & \textbf{0.1538} (.0069) & 8.5609 (6.419) &
		\\ 
		&& Nmse  & 0.4944 (.1184) & \textbf{0.4943} (.1180) & 0.4944 (.1184) &
		\\ \bottomrule\bottomrule
	\end{tabular}
	\label{tb_imc_model2}
\end{table}

\clearpage
\appendix
\section{Appendix: proofs}
\label{sc_appendix_proof}

The main technique for our proofs is the oracle type PAC-Bayes bounds, in the spirit of ~\cite{catoni2004statistical}. We start with a few preliminary lemmas.

\subsection{Preliminary lemmas}

First, we state a version of Bernstein's inequality from Proposition 2.9 page 24 in \cite{MR2319879}.

\begin{lemma}[Bernstein's inequality]
\label{lemmemassart} Let $U_{1}$, \ldots, $U_{n}$ be independent real
valued random variables. Let us assume that there are two constants
$v$ and $w$ such that
$
 \sum_{i=1}^{n} \mathbb{E}[U_{i}^{2}] \leq v 
$
and for all integers $k\geq 3$,
$
 \sum_{i=1}^{n} \mathbb{E}\left[(U_{i})^{k}\right] \leq v\frac{k!w^{k-2}}{2}. 
$
Then, for any $\zeta\in (0,1/w)$,
$$ \mathbb{E}
\exp\left[\zeta\sum_{i=1}^{n}\left[U_{i}-\mathbb{E}(U_{i})\right]
\right]
        \leq \exp\left(\frac{v\zeta^{2}}{2(1-w\zeta)} \right) .$$
\end{lemma}
Another basic tool to derive PAC-Bayes bounds is Donsker and Varadhan's variational inequality, that we state here, see Lemma 1.1.3 in Catoni \cite{catonibook} for a proof (among others). From now, for any $\Theta\subset \mathbb{R}^{n_1}$ or $\Theta\subset \mathbb{R}^{n_1 \times n_2}$, we let $\mathcal{P}(\Theta)$ denote the set of all probability distributions on $\Theta$ equipped with the Borel $\sigma$-algebra. 
We remind that when $(\mu,\nu)\in \mathcal{P}(\Theta)^2$, the Kullback-Leibler divergence is defined by $\mathcal{K}(\nu,\mu) = \int \log \left( \frac{{\rm d}\nu}{{\rm d}\mu}(\theta) \right) \nu({\rm d}\theta) $ if $\nu$ admits a density $\frac{{\rm d}\nu}{{\rm d}\mu}$ with respect to $\mu$, and $\mathcal{K}(\nu,\mu)=+\infty$ otherwise.
\begin{lemma}[Donsker and Varadhan's variational formula]
\label{lemma:dv}
Let $\mu \in\mathcal{P}(\Theta)$. For any measurable, bounded function $h:\Theta\rightarrow\mathbb{R}$ we have:
\begin{equation*}
\log \int {\rm e}^{h(\theta)} \mu({\rm d}\theta) = \sup_{\rho\in\mathcal{P}(\Theta)}\left[\int h(\theta) \rho({\rm d}\theta) -KL (\rho\|\mu)\right].
\end{equation*}
Moreover, the supremum with respect to $\rho$ in the right-hand side is
reached for the Gibbs measure
$\mu_{h}$ defined by its density with respect to $\pi$
\begin{equation}
\label{equa:def:gibbs:measure}
\frac{{\rm d}\mu_{h}}{{\rm d}\mu}(\theta) =  \frac{{\rm e}^{h(\theta)}}
{ \int {\rm e}^{h(\vartheta)} \mu({\rm d}\vartheta) }.
\end{equation}
\end{lemma}
These two lemmas are the only tools we need to prove Theorem \ref{thrm_blr} and Theorem \ref{thrm_contraction_blr}. Their proofs are quite similar, with a few differences. For sake of simplicity, we will state the common parts of the proofs as a separate result in Lemma \ref{lemma:exponential_blr}. Note that the proof of this lemma will use Lemmas \ref{lemmemassart} and \ref{lemma:dv}.

\begin{lemma}
 \label{lemma:exponential_blr}
 Under Assumptions~\ref{assum_bounded} and~\ref{assum_noise1_blr}, put
 \begin{equation}
\label{defalpha_blr}
\alpha 
= 
\left(\lambda -\frac{\lambda^{2} C_1 }{2nq(1-\frac{ C_2 \lambda}{nq})}\right)
\quad
\text{ and }
\quad
\beta 
= 
\left(\lambda +\frac{\lambda^{2} C_1}{2nq(1-\frac{ C_2 \lambda}{nq})}\right) .
\end{equation}
Then for any $\varepsilon\in(0,1)$, and $\lambda \in (0,nq/C_2)$,
\begin{multline}
  \mathbb{E} \int \exp  \Biggl\{  \alpha    \Bigl( R(M) - R(M^*) \Bigr)
 +\lambda\Bigl( -r(M) + r(M^*) \Bigr)    
    - \log \left[\frac{d\widehat{\rho}_{\lambda}}{d \pi} (M)  \right]
         - \log\frac{2}{\varepsilon}
        \Biggr\}
         \widehat{\rho}_{\lambda}(d M)
\leq \frac{\varepsilon}{2}
\label{lemma:exponential:1_blr}
\end{multline}
and
\begin{align}
\mathbb{E} \sup_{\rho \in \mathcal{P}(\mathbb{R}^{p\times k}) } \exp\Biggl[  \beta
               \left(-\int Rd\rho + R(M^*) \right)
+ \lambda \left( \int r d\rho - r(M^*) \right) 
-
\mathcal{K}(\rho, \pi) - \log \frac{2}{\varepsilon}\Biggr] \leq
\frac{\varepsilon}{2}.
\label{lemma:exponential:2_blr}
\end{align}
\end{lemma}

\begin{proof}[\textbf{Proof of Lemma~\ref{lemma:exponential_blr}}]
We prove the first inequality \eqref{lemma:exponential:1} as follows. Fix any $M$ with $\|XMZ \|_\infty\leq C $ and put
 $$ 
 T_{ij} 
 =   
 \left( Y_{ij} - (XM^*Z)_{ij} \right)^{2}
 - 
 \left( Y_{ij} - (\Pi_C(XMZ))_{ij} \right)^{2} .
 $$
Note that the random variables $T_{ij} $ with $ i = 1,\hdots,n; j = 1,\hdots,q $ are independent by construction. We have
\begin{align*}
\sum_{i=1}^{n} \sum_{j=1}^{q} \mathbb{E}[T_{ij}^{2}]  
&  = \sum_{i=1}^{n} \sum_{j=1}^{q}  \mathbb{E}
\left[
       \left( 2Y_{ij} - (XM^*Z)_{ij} - (XMZ)_{ij} \right)^{2}
\left( (XM^*Z)_{ij} - (XMZ)_{ij} \right)^2
            \right]
\\
&  = \sum_{i=1}^{n} \sum_{j=1}^{q}  \mathbb{E} \left[
       \left( 2 E_{ij} + (XM^*Z)_{ij}  - (XMZ)_{ij}  \right)^{2}
\left( (XM^*Z)_{ij} - (XMZ)_{ij} \right)^2
            \right]
\\
&    \leq \sum_{i=1}^{n} \sum_{j=1}^{q}  \mathbb{E} \left[
        8 \left[ E_{ij}^{2} + C^2 \right]
\left[ (XM^*Z)_{ij} -  (XMZ)_{ij} \right]^2
            \right]
\\
&     = \sum_{i=1}^{n} \sum_{j=1}^{q}   8   \mathbb{E}  \left[ E_{ij}^{2} + C^2 \right]
   \left[ (XM^*Z)_{ij} - (XMZ)_{ij} \right]^2
\\
 &   \leq 8 nq( \sigma^2 + C^2 )
\left[ R(M) - R(M^*)\right]
\\
& = nq C_1\left[ R(M) - R(M^*)\right]  =:v(M, M^*).
\end{align*}
Next we have, for any integer $k\geq 3$, that
\begin{align*}
\sum_{i=1}^{n} \sum_{j=1}^{q}  \mathbb{E}\left[(T_{ij})^{k}\right] \leq 
&
\sum_{i=1}^{n} \sum_{j=1}^{q}  \mathbb{E} \left[
       \left| 2Y_{ij} - (XM^*Z)_{ij} - (XMZ)_{ij} \right|^{k}
\left| (XM^*Z)_{ij} - (XMZ)_{ij} \right|^k
            \right]
\\
\leq & \sum_{i=1}^{n} \sum_{j=1}^{q}  \mathbb{E} \left[
       2^{k-1} \left[  |2 E_{ij}|^{k} + (2 C)^k  \right]
 \left| (XM^*Z)_{ij} - (XMZ)_{ij} \right|^{k}
            \right]
\\
\leq  &  \sum_{i=1}^{n} \sum_{j=1}^{q}  \mathbb{E} \left[
       2^{2k-1}\left(  |E_{ij}|^{k} + C^k \right)
       (2C)^{k-2}
 \left| (XM^*Z)_{ij} - (XMZ)_{ij} \right|^{2}
            \right]
\\
\leq    &      2^{2k-1} \left[ \sigma^{2}k!\xi^{k-2}
             + C^k   \right]  (2C)^{k-2}
\sum_{i=1}^{n} \sum_{j=1}^{q} \mathbb{E}  
\left| (XM^*Z)_{ij} - (XMZ)_{ij} \right|^{2}
\\
\leq  &  
\frac{ 2^{3k-3} \left[ \sigma^{2}k!\xi^{k-2} + C^k \right] C^{k-2} }{8(\sigma^2 + C^2) } v(M, M^*)
\\
\leq  &  \frac{ 2^{3k-6}  \left[ \sigma^{2}\xi^{k-2} + C^k \right] C^{k-2}}{(\sigma^2 + C^2) } k! v(M, M^*)
\\
\leq  &  2^{3k-5} \left[ \xi^{k-2} + C^{k-2} \right] C^{k-2} k! v(M, M^*)
\\
\leq  &  2^{3k-4} \max(\xi,C)^{k-2} C^{k-2} k! v(M, M^*)
=    
[2^3 \max(\xi,C) C ]^{k-2} 2^2 k! v(M, M^*)
\end{align*}
and use the fact that, for any $k\geq 3$, $2^2 \leq 2^{3(k-2)}/2 $ to obtain
\begin{align*}
\sum_{i=1}^{n} \sum_{j=1}^{q}  \mathbb{E}\left[(T_{ij})^{k}\right] 
\leq   
 \frac{[2^6 \max(\xi,C) C ]^{k-2} k! v(M, M^*)}{2}
= 
v(M, M^*) \frac{k!C_2^{k-2}}{2}.
\end{align*}
Thus, we can apply Lemma~\ref{lemmemassart} with $U_i := T_i$, $v := v(M, M^*)$, $w:=C_2$ and $\zeta := \lambda/nq$. We obtain, for any $\lambda\in(0,nq/w)=(0,nq/C_2)$,
\begin{align*}
\mathbb{E} \exp\left[\lambda
\Bigl( R(M)-R(M^*)-r(M)+r(M^*)\Bigr)\right]
 \leq
\exp\left[\frac{v\lambda^{2}}{2(nq)^{2}(1-\frac{w\lambda}{nq})}\right]
 =
\exp\left[\frac{ C_1
\left[ R(M) - R(M^*)\right] \lambda^{2}}{2nq(1-\frac{C_2 \lambda}{nq})}\right].
\end{align*}
Rearranging terms, and using the definition of $\alpha$ in \eqref{defalpha_blr},
$$
\mathbb{E} \exp\left[\alpha
                      \Bigl( R(M) - R(M^*) \Bigr)
                    +\lambda\Bigl( -r(M) + r(M^*) \Bigr) \right] \leq 1.
$$
Multiplying both sides by $\varepsilon/2$ and then integrating w.r.t. the probability distribution $ \pi(.) $, we get
$$
\int \mathbb{E} \exp\Biggl[\alpha
                      \Bigl( R(M) - R(M^*) \Bigr)
                    +\lambda\Bigl( -r(M) + r(M^*) \Bigr)
         - \log\frac{2}{\varepsilon}\Biggr]  \pi (d M) \leq \frac{\varepsilon}{2}.
$$
Next, Fubini's theorem gives
$$
\mathbb{E} \int  \exp\Biggl[\alpha
                      \Bigl( R(M) - R(M^*) \Bigr)
                    +\lambda\Bigl( -r(M) + r(M^*) \Bigr)
         - \log\frac{2}{\varepsilon}\Biggr]  \pi (d M) \leq \frac{\varepsilon}{2}.
$$
and note that for any measurable function $h$,
\begin{align*}
\int  \exp[h(M)]  \pi (d M)
= 
\int  \exp\left[h(M)- \log \frac{{\rm d} \widehat{\rho}_{\lambda}}{{\rm d}\pi}(M) \right] \widehat{\rho}_{\lambda}(d M)
\end{align*}
to get \eqref{lemma:exponential:1_blr}.

Let us now prove \eqref{lemma:exponential:2_blr}. Here again, we start with an application of Lemma~\ref{lemmemassart}, but this time with $U_i := - T_i$ (we keep $v := v(M, M^*)$, $w:=C_2$ and $\zeta := \lambda/nq $). We obtain, for any $\lambda\in(0,nq/C_2) $,
$$
\mathbb{E}  \exp\left[\lambda
\Bigl( r(M)+r(M^*)-R(M)+R(M^*)\Bigr)\right]
\leq
\exp\left[\frac{ C_1
\left[ R(M) - R(M^*)\right] \lambda^{2}}{2nq(1-\frac{C_2 \lambda}{nq})}\right].
$$
Rearranging terms, using the definition of $\beta$ in~\eqref{defalpha_blr} and multiplying both sides by $\varepsilon/2$, we obtain
\begin{equation*}
\mathbb{E} \exp\Biggl[\beta
               \left(-R(M) + R(M^*) \right)
+ \lambda \left( r(M)- r(M^*) \right) - \log \frac{2}{\varepsilon}\Biggr] \leq
\frac{\varepsilon}{2}.
\end{equation*}
We integrate with respect to $\pi$ and use Fubini to get:
\begin{equation*}
\mathbb{E} \int \exp\Biggl[\beta
               \left(-R(M) + R(M^*) \right)
+ \lambda \left( r(M)- r(M^*) \right) - \log \frac{2}{\varepsilon}\Biggr] \pi({\rm d} M ) \leq
\frac{\varepsilon}{2}.
\end{equation*}
Here, we use a different argument from the proof of the first inequality: we use Lemma \ref{lemma:dv} on the integral, this gives directly \eqref{lemma:exponential:2_blr}.
\end{proof}
Finally, in both proofs, we will use quite often distributions $\rho\in\mathcal{P}(\mathbb{R}^{p\times k})$ that will be defined as translations of the prior $\pi$. We introduce the following notation.
\begin{dfn}
\label{dfn:posterior:transla}
 For any  matrix $\bar{M} \in \mathbb{R}^{p\times k}$, we define $\rho_{\bar{M}} \in \mathcal{P}(\mathbb{R}^{p\times k})$
 by 
 $$
 \rho_{\bar{M}}(M) = \pi(\bar{M}-M).
 $$
\end{dfn}

\noindent The following technical lemmas from \cite{dalalyan2020exponential} will be useful in the proofs.

\begin{lemma}[Lemma 1 in \cite{dalalyan2020exponential}]
\label{lemma:arnak:1}
 We have
 $ \int \|M\|_{F}^2 \pi({\rm d}M) \leq  pk \tau^2. $
\end{lemma}

\begin{lemma}[Lemma 2 in \cite{dalalyan2020exponential}]
\label{lemma:arnak:2}
For any $\bar{M} \in \mathbb{R}^{p\times k}$, we have
$$
 \mathcal{K}( \rho_{\bar{M}} ,  \pi) 
\leq 
2 {\rm rank} (\bar{M}) (k+p+2) 
\log \left( 1+ \frac{\| \bar{ M } \|_F}{\tau \sqrt{2{\rm rank} (\bar{M})}} \right) $$
with the convention $0\log(1+0/0)=0$.
\end{lemma}

\subsection{Proof of Theorem \ref{thrm_blr}}

\begin{proof}[\textbf{Proof of Theorem \ref{thrm_blr}}]
An application of Jensen's inequality on  inequality~\eqref{lemma:exponential:1_blr} yields
	$$
	\mathbb{E} \exp\Biggl[\alpha
	\left( \int R d \widehat{\rho}_{\lambda} - R(M^*) \right)
	+\lambda\left( -\int r d \widehat{\rho}_{\lambda} + r(M^*) \right)
	- \mathcal{K}(\widehat{\rho}_{\lambda}, \pi)
	- \log\frac{2}{\varepsilon}\Biggr] \leq \frac{\varepsilon}{2}.
	$$
	Using the standard Chernoff's trick to transform an exponential moment inequality into a deviation inequality, that is: $\exp(x) \geq \mathbf{1}_{\mathbb{R}_{+}}(x)$, we obtain
	\begin{align}
	\mathbb{P}\Biggl\{ \Biggr[ \alpha
	\left(\int R d\widehat{\rho}_{\lambda} - R(M^*)\right)
	+\lambda\left(-\int r d\widehat{\rho}_{\lambda} + r(M^*) \right)
	- \mathcal{K}(\widehat{\rho}_{\lambda}, \pi)
	- \log\frac{2}{\varepsilon}\Biggr] \geq 0
	\Biggr\} 
	\leq
	\frac{\varepsilon}{2}
	\label{equa:step1_blr}
	\end{align}
	Using~\eqref{pyta} we have
	\begin{align*}
	\int R d\widehat{\rho}_{\lambda} - R(M^*) 
	=
	\dfrac{1}{nq} \int \left\|\Pi_C(XMZ) -XM^*Z \right\|_{F}^2\widehat{\rho}_{\lambda}({\rm d}M)
	&  \geq   
	\dfrac{1}{nq} \left\|\int \Pi_C(XMZ)\widehat{\rho}_{\lambda}({\rm d}M) -XM^*Z \right\|_{F}^2
	\\
	& \geq   
	\dfrac{1}{nq} \left\|\widehat{XMZ}_{\lambda} -XM^*Z \right\|_{F}^2
	\end{align*}
	where we used Jensen's inequality in the first line, and the definition of $\widehat{XMZ}_{\lambda}$ from the first to the second line. Plugging this into our probability bound~\eqref{equa:step1_blr}, and dividing both sides by $\alpha$, we obtain
	\begin{equation*}
	\mathbb{P}\Biggl\{ 
	\dfrac{1}{nq} \left\|\widehat{XMZ}_{\lambda} -XM^*Z \right\|_{F}^2
	\leq 
	\frac{ \int r d\widehat{\rho}_{\lambda} - r(M^*) +
		\frac{1}{\lambda}\left[\mathcal{K}(\widehat{\rho}_{\lambda}, \pi) 
		+ \log\frac{2}{\varepsilon}\right] } {\frac{\alpha}{\lambda} }
	\Biggr\}
	\geq 
	1-\frac{\varepsilon}{2}
	\end{equation*}
	under the additional condition that $\lambda$ is such that $\alpha>0$, that we will assume from now (note that this is satisfied by $\lambda^*$). Using Lemma~\ref{lemma:dv} we can rewrite this as
	\begin{equation}
	\label{interm3bis_blr}
	\mathbb{P}\Biggl\{ 
	\dfrac{1}{nq} \left\|\widehat{XMZ}_{\lambda} -XM^*Z \right\|_{F}^2
	\leq 
	\inf_{\rho \in \mathcal{P}(\mathbb{R}^{p\times k}) } \frac{ \int r
		d\rho - r(M^*) 
		+
		\frac{1}{\lambda}\left[\mathcal{K}(\rho, \pi)
		+ 
		\log\frac{2}{\varepsilon}\right] } {\frac{\alpha}{\lambda} } \Biggr\} \geq 1-\frac{\varepsilon}{2}.
	\end{equation}
	
	We consider now the consequences of the second inequality in Lemma~\ref{lemma:exponential_blr}, that is~\eqref{lemma:exponential:2_blr}. Chernoff's trick and rearranging terms a little, we get
	\begin{align*}
	\mathbb{P}\Biggl\{ \forall \rho \in\mathcal{P}(\mathbb{R}^{p\times k}), \int rd\rho - r(M^*)
	\leq 
	\frac{\beta}{\lambda} \left[\int
	Rd\rho - R(M^*) \right] + \frac{1}{\lambda}\left[
	\mathcal{K}(\rho, \pi) + \log \frac{2}{\varepsilon} \right]
	\Biggr\}\geq 1 - \frac{\varepsilon}{2}.
	\end{align*}
	which we can rewrite as, $ \forall \rho \in\mathcal{P}(\mathbb{R}^{p\times k}), $ with probability at least $ 1 - \frac{\varepsilon}{2} $,
	\begin{align}
 \int rd\rho - r(M^*)
	\leq 
	\frac{\beta}{\lambda} 
	\int \dfrac{1}{nq}  \|\Pi_C(XMZ)-X M^*Z \|_{F}^2 \rho({\rm d}M) 
	+ 
	\frac{1}{\lambda}\left[	\mathcal{K}(\rho, \pi) 
	+ 
	\log \frac{2}{\varepsilon} \right]
 .
	\label{interm4_blr}
	\end{align}
	\\
	Combining (\ref{interm4_blr}) and (\ref{interm3bis_blr}) with a union bound
	argument gives the following bound, with probability at least $ 1-\varepsilon $,
	\begin{align*}
\frac{1}{nq}  \left\|\widehat{XMZ}_{\lambda} -XM^*Z \right\|_{F}^2
	\leq 
	\inf_{\rho \in \mathcal{P}(\mathbb{R}^{p\times k}) } 
	\frac{ \beta  \int \frac{1}{nq}  \|\Pi_C(XMZ)-X M^*Z \|_{F}^2 
		\rho({\rm d}M) + 2 \left[
		\mathcal{K}(\rho, \pi) + \log \frac{2}{\varepsilon} \right] } {
		\alpha  } 
	.
	\end{align*}
	Noting that, for any $(i,j)$, $(X M^*Z)_{i,j} \in [-C,C]$ implies
	that $| (\Pi_C(X MZ))_{i,j} - (X M^*Z)_{i,j} | \leq | (X MZ)_{i,j} - (X M^*Z)_{i,j} |$ and thus
	\begin{align*}
\frac{1}{nq}  \left\|\widehat{XMZ}_{\lambda} -XM^*Z \right\|_{F}^2
\leq 
\inf_{\rho \in \mathcal{P}(\mathbb{R}^{p\times k}) } 
\frac{ \beta  \int \frac{1}{nq}  \| XMZ -X M^*Z \|_{F}^2 
	\rho({\rm d}M) + 2 \left[
	\mathcal{K}(\rho, \pi) + \log \frac{2}{\varepsilon} \right] } {
	\alpha  } 
.
\end{align*}
	The end of the proof consists in making the right-hand side in the inequality more explicit. In order to do so, we restrict the infimum bound above to the distributions given by Definition~\ref{dfn:posterior:transla}.
	\begin{multline}
	\mathbb{P}\Biggl\{ \frac{1}{nq} \left\|\widehat{XMZ}_{\lambda} -XM^*Z \right\|_{F}^2
	\leq 
	\inf_{\bar{M}\in\mathbb{R}^{p\times k} } \frac{ \beta  \int \frac{1}{nq}\|XMZ -X M^*Z \|_{F}^2 \rho_{\bar{M}} ({\rm d}M) + 2 \left[
		\mathcal{K}(\rho_{\bar{M}}, \pi) + \log \frac{2}{\varepsilon} \right] } {
		\alpha  } \Biggr\} 
	\geq 
	1-\varepsilon.
	\label{PAC_bound_blr}
	\end{multline}
	We see immediately that Dalalyan's lemma will be extremely useful for that. First, Lemma~\ref{lemma:arnak:2} provides an upper bound on $\mathcal{K}(\rho_{\bar{M}}, \pi)$. Moreover,
	\begin{align*}
	& \int  \|XMZ - X M^*Z \|_{F}^2 \rho_{\bar{M}} ({\rm d}M)
	\\
	& \leq  
	\int \|X\bar{M}Z - X M^*Z - XMZ \|_{F}^2 \pi({\rm d}M)
	\\
	& =   \|X\bar{M}Z - X M^*Z  \|_{F}^2
	-
	2 \int \sum_{i,j} (X\bar{M}Z - X M^*Z)_{j,i} (XMZ)_{i,j} \pi({\rm d}M)
	+ 
	\int \|XMZ \|_{F}^2 \pi({\rm d}M).
	\end{align*}
	The second term in the right-hand side is null because $\pi$ is centered, and thus
	\begin{align*}
	\int \|XMZ - X M^*Z \|_{F}^2 \rho_{\bar{M}} ({\rm d}M)
	& \leq   
	\|X\bar{M}Z - X M^*Z \|_{F}^2 +  \int \|X M Z \|_{F}^2 \, \pi({\rm d}M)
	\\
	& \leq \|X\bar{M}Z - X M^*Z \|_{F}^2 
	+  \|X \|_{F}^2 \|Z\|_{F}^2
	\int \|M \|_{F}^2 \, \pi({\rm d}M)
	\\
	& \leq \|X\bar{M}Z - X M^*Z \|_{F}^2 
	+  \|X \|_{F}^2\|Z \|_{F}^2 pk \tau^2
	\end{align*}
	where we used elementary properties of the Frobenius norm, and Lemma~\ref{lemma:arnak:1} in the last line. We can now plug this (and Lemma~\ref{lemma:arnak:2}) back into~\eqref{PAC_bound_blr} to get:
	\begin{multline*}
	\mathbb{P}\Biggl\{ \frac{1}{nq} \left\|\widehat{XMZ}_{\lambda} -XM^*Z \right\|_{F}^2
	\leq 
	\inf_{\bar{M}\in\mathbb{R}^{p\times k} } \Biggl[ \frac{\beta}{\alpha} \frac{1}{nq}\|X\bar{M}Z - X M^*Z \|_{F}^2 
	+ 
	\frac{\beta}{\alpha} \frac{1}{nq}\|X \|_{F}^2\|Z \|_{F}^2 pk \tau^2
	\\
	+ 
	\frac{1}{\alpha} \left( 4 {\rm rank} (\bar{M}) (k+p+2) \log \left( 1+ \frac{\| \bar{ M } \|_F}{\tau \sqrt{2{\rm rank} (\bar{M})}} \right) + 2 \log \frac{2}{\varepsilon} \right) \Biggr] \Biggr\} 
	\geq 
	1-\varepsilon.
	\end{multline*}
	We are now making the constants be explicit.
	First, if $\lambda \leq nq/(2C_2)$, then $2nq( 1 - C_2 \lambda/nq ) \geq np$ and thus
	$$
	\frac{\beta}{\alpha}
	= 
	\frac{1+ \frac{\lambda C_1 }{2nq(1-\frac{ C_2 \lambda}{nq})}}{1- \frac{\lambda C_1 }{2nq(1-\frac{ C_2 \lambda}{nq})}}
	\leq
	\frac{1+ \frac{\lambda C_1 }{nq} }{1- \frac{\lambda C_1 }{nq} }.
	$$
	Then, $ \lambda \leq \frac{nq \delta}{C_1(1+\delta)} $ leads to
	$
	\frac{\beta}{\alpha} \leq (1+\delta).
	$
	
	Note that $\lambda^* = nq \min(1/(2C_2), \delta/[C_1(1+\delta)] )$ satisfies these two conditions, so from now $\lambda = \lambda^*$. We also use the following:
	\begin{align*}
	\frac{1}{\alpha} 
	= 
	\frac{1}{\lambda^*\left(1- \frac{\lambda^* C_1 }{2nq( 1- C_2 \lambda^* / nq )}\right)}
	\leq 
	\frac{\beta}{\lambda^* \alpha} 
	\leq 
	\frac{(1+\delta)}{nq \min(1/(2C_2), \delta/[C_1(1+\delta)] )} 
	\leq 
	\frac{C_1 (1+\delta)^2 }{nq \delta}.
	\end{align*}
	So far the bound is: 
	\begin{multline*}
	\mathbb{P}\Biggl\{ \frac{1}{nq} \left\|
	\widehat{XMZ}_{\lambda^*} -XM^*Z \right\|_{F}^2
	\leq 
	\inf_{\bar{M}\in\mathbb{R}^{p\times k} } \Biggl[ \frac{(1+\delta)}{nq}  \|X\bar{M}Z - X M^*Z \|_{F}^2 
	+
	\frac{(1+\delta)}{nq} \|X \|_{F}^2\|Z \|_{F}^2 p k\tau^2
	+
	\\
	\frac{C_1 (1+\delta)^2 \left( 4 {\rm rank} (\bar{M}) (k+p+2) \log \left( 1+ \frac{\| \bar{ M } \|_F}{\tau \sqrt{2{\rm rank} (\bar{M})}} \right) + 2 \log \frac{2}{\varepsilon} \right)}{ nq \delta } \Biggr] \Biggr\} 
	\geq 
	1-\varepsilon.
	\end{multline*}
	In particular, with probability at least $ 1-\varepsilon $, the choice $\tau^2 = C_1 (k+p)/(n k q \|X \|_{F}^2 \|Z \|_{F}^2 ) $ gives
	\begin{multline*}
	 \frac{1}{nq} \left\|\widehat{XMZ}_{\lambda^*} -XM^*Z \right\|_{F}^2
	\leq 
	\inf_{\bar{M}\in\mathbb{R}^{p\times k} } \Biggl[ \frac{(1+\delta)}{nq} \|X\bar{M}Z - X M^*Z \|_{F}^2 + \frac{C_1 (1+\delta) (k+p)}{nq}
	+
	\\
	\frac{C_1 (1+\delta)^2 \left( 4 {\rm rank} (\bar{M}) (k+p+2) 
	\log 
	\left( 1+\frac{\| X \|_F \|Z \|_{F} \| \bar{ M } \|_F}{ \sqrt{C_1}} \sqrt{ \frac{n k q}{(k+p) {\rm rank}(\bar{M})}} \right) + 2 \log \frac{2}{\varepsilon} \right)}{ nq\delta } \Biggr].
	\end{multline*}
\end{proof}

\subsection{Proof of theorem \ref{thrm_contraction_blr}}
\label{sc_proof_thrm_contraction_blr}
\begin{proof}[\textbf{Proof of Theorem \ref{thrm_contraction_blr}}]
	We also start with an application of Lemma~\ref{lemma:exponential_blr}, and focus on~\eqref{lemma:exponential:1_blr}, applied to $\varepsilon:=\varepsilon_n$, that is:
	\begin{align*}
	\mathbb{E} \int \exp  \Biggl\{   \alpha    \Bigl( R(M) - R(M^*) \Bigr)
	+\lambda\Bigl( -r(M) + r(M^*) \Bigr)               - \log \left[\frac{d\widehat{\rho}_{\lambda}}{d \pi} (M)  \right]
	- \log\frac{2}{\varepsilon_n}
	\Biggr\}
	\widehat{\rho}_{\lambda}(d M)
	\leq \frac{\varepsilon_n}{2}.
	\end{align*}
	Using Chernoff's trick, this gives:
	$$
	\mathbb{E} \Bigl[ \mathbb{P}_{M\sim \widehat{\rho}_{\lambda}} (M\in\mathcal{A}_n) \Bigr]
	\geq 1-\frac{\varepsilon_n}{2}
	$$
	where
	$$
	\mathcal{A}_n = \left\{M: \alpha    \Bigl( R(M) - R(M^*) \Bigr)
	+\lambda\Bigl( -r(M) + r(M^*) \Bigr)      \leq      \log \left[\frac{d\widehat{\rho}_{\lambda}}{d \pi} (M)  \right]
	+ \log\frac{2}{\varepsilon_n} \right\}.
	$$
	Using the definition of $\widehat{\rho}_\lambda $, for $M\in \mathcal{A}_n$ we have
	\begin{align*}
	\alpha    \Bigl( R(M) - R(M^*) \Bigr)
	&    \leq  \lambda\Bigl( r(M) - r(M^*) \Bigr)  +       \log \left[\frac{d\widehat{\rho}_{\lambda}}{d \pi} (M)  \right]
	+ \log\frac{2}{\varepsilon_n}
	\\
	& \leq -\log\int\exp\left[-\lambda r(M)\right]\pi({\rm d}M) - \lambda r(M^*)
	+ \log\frac{2}{\varepsilon_n}
	\\
	& = \lambda\Bigl( \int r(M) \widehat{\rho}_{\lambda}({\rm d}M) - r(M^*) \Bigr)  +    \mathcal{K}(\widehat{\rho}_\lambda,\pi)
	+ \log\frac{2}{\varepsilon_n}
	\\
	& = \inf_{\rho} \left\{ \lambda\Bigl( \int r(M) \rho({\rm d}M) - r(M^*) \Bigr)  +    \mathcal{K}(\rho,\pi)
	+ \log\frac{2}{\varepsilon_n} \right\}.
	\end{align*}
	
	\noindent Now, let us define
	$$ \mathcal{B}_n = \left\{\forall\rho\text{, }\beta \left(-\int Rd\rho + R(M^*) \right)
	+ \lambda \left( \int r d\rho - r(M^*) \right) \leq
	\mathcal{K}(\rho, \pi) + \log \frac{2}{\varepsilon_n}\right\}. $$
	
	\noindent Using~\eqref{lemma:exponential:2_blr}, we have that
	$$
	\mathbb{E} \Bigl[\mathbf{1}_{\mathcal{B}_n} \Bigr]
	\geq 1-\frac{\varepsilon_n}{2}.
	$$
	We will now prove that, if $\lambda$ is such that $\alpha>0$,
	$$
	\mathbb{E} \Bigl[ \mathbb{P}_{M\sim \widehat{\rho}_{\lambda}} (M\in\mathcal{M}_n) \Bigr] \geq 
	\mathbb{E} \Bigl[ \mathbb{P}_{M\sim \widehat{\rho}_{\lambda}} (M\in\mathcal{A}_n)\mathbf{1}_{\mathcal{B}_n} \Bigr]
	$$
	which, together with
	\begin{align*}
	\mathbb{E} \Bigl[ \mathbb{P}_{M\sim \widehat{\rho}_{\lambda}} (M\in\mathcal{A}_n)\mathbf{1}_{\mathcal{B}_n} \Bigr]
	& = \mathbb{E} \Bigl[ (1-\mathbb{P}_{M\sim \widehat{\rho}_{\lambda}} (M\notin\mathcal{A}_n)) (1-\mathbf{1}_{\mathcal{B}^c_n})\Bigr]
	\\
	& \geq \mathbb{E} \Bigl[ 1-\mathbb{P}_{M\sim \widehat{\rho}_{\lambda}} (M\notin\mathcal{A}_n) - \mathbf{1}_{\mathcal{B}^c_n}
	\Bigr]
	\\
	& \geq 1-\varepsilon_n
	\end{align*}
	will bring
	\begin{equation*}
	\mathbb{E} \Bigl[ \mathbb{P}_{M\sim \widehat{\rho}_{\lambda}} (M\in\mathcal{M}_n) \Bigr] \geq 1-\varepsilon_n.
	\end{equation*}
	In order to do so, assume that we are on the set $\mathcal{B}_n$, and let $M\in\mathcal{A}_n$. Then,
	\begin{align*}
	\alpha    \Bigl( R(M) - R(M^*) \Bigr)
	& \leq \inf_{\rho} \left\{ \lambda\Bigl( \int r(M) \rho({\rm d}M) - r(M^*) \Bigr)  +    \mathcal{K}(\rho,\pi)
	+ \log\frac{2}{\varepsilon_n} \right\}
	\\
	& \leq \inf_{\rho} \left\{ \beta \Bigl( \int R(M) \rho({\rm d}M) - R(M^*) \Bigr)  +   2 \mathcal{K}(\rho,\pi)
	+ 2 \log\frac{2}{\varepsilon_n} \right\}
	\end{align*}
	that is,
	$$
	R(M) - R(M^*) \leq \inf_{\rho \in \mathcal{P}(\mathbb{R}^{p\times k}) } \frac{ \beta \left[\int Rd\rho -
		R(M^*) \right] + 2 \left[
		\mathcal{K}(\rho, \pi) + \log \frac{2}{\varepsilon} \right] } {
		\alpha  }
	$$
	or, rewriting it in terms of norms,
	$$
	\left\|\Pi_C(XMZ) -XM^*Z \right\|_{F}^2
	\\
	\leq 
	\inf_{\bar{M}\in\mathbb{R}^{p\times k} } \frac{ \beta  \int \|XMZ - X M^*Z \|_{F}^2 \rho_{\bar{M}} ({\rm d}M) + 2 \left[
		\mathcal{K}(\rho_{\bar{M}}, \pi) + \log \frac{2}{\varepsilon} \right] } {
		\alpha  }.
	$$
	We upper-bound the right-hand side exactly as in the proof of Theorem~\ref{thrm_blr}, this gives
	$M\in\mathcal{M}_n$.
	
\end{proof}

\subsection{Proof of Theorem \ref{thrm_main}}

\begin{lemma}
	\label{lemma:exponential}
	Under Assumptions~\ref{assum_bounded} and~\ref{assum_noise}, put
	\begin{equation}
	\label{defalpha}
	\alpha' = \left(\lambda
	-\frac{\lambda^{2} C'_1 }{2m(1-\frac{ C'_2 \lambda}{m})}\right)
\quad \text{ and } \quad
	\beta' = \left(\lambda
	+\frac{\lambda^{2} C'_1}{2m(1-\frac{ C'_2 \lambda}{m})}\right) .
	\end{equation}
	Then for any $\varepsilon\in(0,1)$, and $\lambda \in (0,m/C'_2)$,
	\begin{multline}
	\mathbb{E} \int \exp  \Biggl\{  \alpha' \Bigl( R'(M) - R'(M^*) \Bigr)
	+\lambda\Bigl( -r'(M) + r'(M^*) \Bigr)    
	- \log \left[\frac{d\widehat{\rho}'_{\lambda}}{d \pi} (M)  \right]
	- \log\frac{2}{\varepsilon}
	\Biggr\}
	\widehat{\rho}'_{\lambda}(d M)
	\leq \frac{\varepsilon}{2}
	\label{lemma:exponential:1}
	\end{multline}
	and
	\begin{align}
	\mathbb{E} \sup_{\rho \in \mathcal{P}(\mathbb{R}^{p\times k}) } \exp\Biggl[  \beta'
	\left(-\int R'd\rho + R'(M^*) \right)
	+ \lambda \left( \int r' d\rho - r'(M^*) \right) 
	-
	\mathcal{K}(\rho, \pi) - \log \frac{2}{\varepsilon}\Biggr] \leq
	\frac{\varepsilon}{2}.
	\label{lemma:exponential:2}
	\end{align}
\end{lemma}

\begin{proof}[\textbf{Proof of Lemma~\ref{lemma:exponential}}]
The inequality \eqref{lemma:exponential:1} is proved in a similar way to the proof of Lemma \ref{lemma:exponential_blr}. That is we apply Lemma \ref{lemmemassart} to the following independent random variables
	$$ 
	V_{i} 
	=   
	\left( Y_{i} - (XM^*Z)_{i} \right)^{2}
	- 
	\left( Y_{i} - (\Pi_C(XMZ))_{i} \right)^{2} ,i = 1,\hdots,m.
	$$
The proof of the inequality \eqref{lemma:exponential:2} is processed similar in the proof of Lemma \ref{lemma:exponential_blr} in which we apply Lemma \ref{lemmemassart} to the independent random variables $ -V_{i} ,i = 1,\hdots,m. $
\end{proof}

\begin{proof}[\textbf{Proof of Theorem \ref{thrm_main}}]
Similar to the proof of Theorem \ref{thrm_blr}, until the \eqref{interm3bis_blr}, and noting that using~\eqref{pyta_imc} we have
\begin{align*}
 \int R' d\widehat{\rho}_{\lambda} - R'(M^*) 
  =
  \int \left\|\Pi_C(XMZ) -XM^*Z \right\|_{F,\Pi}^2\widehat{\rho}'_{\lambda}({\rm d}M)
 &  \geq   
 \left\|\int \Pi_C(XMZ)\widehat{\rho}'_{\lambda}({\rm d}M) -XM^*Z \right\|_{F,\Pi}^2
  \\
  & \geq   \left\|\widehat{XMZ}_{\lambda} -XM^*Z \right\|_{F,\Pi}^2,
\end{align*}
thus we obtain
\begin{equation}\label{interm3bis}
\mathbb{P}\Biggl\{ 
\left\|\widehat{XMZ}_{\lambda} -XM^*Z \right\|_{F,\Pi}^2
\leq 
\inf_{\rho \in \mathcal{P}(\mathbb{R}^{p\times k}) } 
\frac{ \int r' d\rho - r'(M^*) 
+
\frac{1}{\lambda}\left[\mathcal{K}(\rho, \pi)
+ 
\log\frac{2}{\varepsilon}\right] } {\frac{\alpha'}{\lambda} } \Biggr\} \geq 1-\frac{\varepsilon}{2}.
\end{equation}

We consider now the consequences of inequality \eqref{lemma:exponential:2} in Lemma~\ref{lemma:exponential}. Chernoff's trick and rearranging terms a little, we get, $ \forall \rho \in\mathcal{P}(\mathbb{R}^{p\times k}) $, with probability at least $ 1 - \frac{\varepsilon}{2} $,
\begin{align}
 \int r'd\rho - r'(M^*)
 \leq 
 \frac{\beta'}{\lambda} 
 \int \|\Pi_C(XMZ)-X M^*Z \|_{F,\Pi}^2 \rho({\rm d}M) 
 + 
 \frac{1}{\lambda}\left[
\mathcal{K}(\rho, \pi) + \log \frac{2}{\varepsilon} \right]
 .
\label{interm4}
\end{align}
\\
Combining (\ref{interm4}) and (\ref{interm3bis}) with a union bound argument gives the bound and noting that for any $(i,j)$, $(X M^*Z)_{i,j} \in [-C,C]$ implies
that $| (\Pi_C(X MZ))_{i,j} - (X M^*Z)_{i,j} | \leq | (X MZ)_{i,j} - (X M^*Z)_{i,j} |$ and thus
\begin{align*}
\mathbb{P}\Biggl\{  \left\|\widehat{XMZ}_{\lambda} -XM^*Z \right\|_{F,\Pi}^2
\leq 
\inf_{\rho \in \mathcal{P}(\mathbb{R}^{p\times k}) } 
\frac{ \beta' \int \|XMZ -X M^*Z \|_{F,\Pi}^2 \rho({\rm d}M) + 2 \left[
\mathcal{K}(\rho, \pi) + \log \frac{2}{\varepsilon} \right] } {
\alpha'  } \Biggr\} 
\geq 
1-\varepsilon.
\end{align*}
We are now making the right-hand side in the inequality more explicit. In order to do so, we restrict the infimum bound above to the distributions given by Definition~\ref{dfn:posterior:transla}.
\begin{multline}
\mathbb{P}\Biggl\{  \left\|\widehat{XMZ}_{\lambda} -XM^*Z \right\|_{F,\Pi}^2
\leq 
\inf_{\bar{M}\in\mathbb{R}^{p\times k} } 
\frac{ \beta'  \int \|XMZ -X M^*Z \|_{F,\Pi}^2 \rho_{\bar{M}} ({\rm d}M) + 2 \left[
\mathcal{K}(\rho_{\bar{M}}, \pi) 
+ 
\log \frac{2}{\varepsilon} \right] } {
\alpha'  } \Biggr\} 
\geq 
1-\varepsilon.
\label{PAC_bound}
\end{multline}
We see immediately that Dalalyan's lemma will be extremely useful for that. First, Lemma~\ref{lemma:arnak:2} provides an upper bound on $\mathcal{K}(\rho_{\bar{M}}, \pi)$.

 Moreover,
\begin{align*}
 & \int  \|XMZ - X M^*Z \|_{F,\Pi}^2 \rho_{\bar{M}} ({\rm d}M)
  \\
  & \leq  
   \int \|X\bar{M}Z - X M^*Z - XMZ \|_{F,\Pi}^2 \pi({\rm d}M)
  \\
  & =   \|X\bar{M}Z - X M^*Z  \|_{F,\Pi}^2
  -
  2 \int \sum_{i,j} \Pi_{i,j} (X\bar{M}Z - X M^*Z)_{j,i} (XMZ)_{i,j} \pi({\rm d}M)
 + 
 \int \|XMZ \|_{F,\Pi}^2 \pi({\rm d}M).
\end{align*}
The second term in the above right-hand side is null because $\pi$ is centered, and thus
\begin{align*}
  \int \|XMZ - X M^*Z \|_{F,\Pi}^2 \rho_{\bar{M}} ({\rm d}M)
  & \leq   \|X\bar{M}Z - X M^*Z \|_{F,\Pi}^2 + \int \|XMZ \|_{F,\Pi}^2 \pi({\rm d}M)
  \\
  & \leq   \|X\bar{M}Z - X M^*Z \|_{F,\Pi}^2 +  \int \|X M Z \|_{F}^2 \, \pi({\rm d}M)
  \\
  & \leq \|X\bar{M}Z - X M^*Z \|_{F,\Pi}^2 
  +  \|X \|_{F}^2 \|Z\|_{F}^2
   \int \|M \|_{F}^2 \, \pi({\rm d}M)
  \\
  & \leq \|X\bar{M}Z - X M^*Z \|_{F,\Pi}^2 
  +  \|X \|_{F}^2\|Z \|_{F}^2 pk \tau^2
\end{align*}
where we used elementary properties of the Frobenius norm, and Lemma~\ref{lemma:arnak:1} in the last line. We can now plug this (and Lemma~\ref{lemma:arnak:2}) back into~\eqref{PAC_bound} to get:
\begin{multline*}
\mathbb{P}\Biggl\{  \left\|\widehat{XMZ}_{\lambda} -XM^*Z \right\|_{F,\Pi}^2
\leq 
\inf_{\bar{M}\in\mathbb{R}^{p\times k} } \Biggl[ \frac{\beta'}{\alpha'} \|X\bar{M}Z - X M^*Z \|_{F,\Pi}^2 
+ 
\frac{\beta'}{\alpha'} \|X \|_{F}^2\|Z \|_{F}^2 pk \tau^2
\\
+ 
 \frac{1}{\alpha'} \left( 4 {\rm rank} (\bar{M}) (k+p+2) \log \left( 1+ \frac{\| \bar{ M } \|_F}{\tau \sqrt{2{\rm rank} (\bar{M})}} \right) + 2 \log \frac{2}{\varepsilon} \right) \Biggr] \Biggr\} 
\geq 
1-\varepsilon.
\end{multline*}
We are now making the constants be explicit.
First, if $\lambda \leq m/(2C'_2)$, then $2m( 1 - C'_2 \lambda/m ) \geq m$ and thus
$$
\frac{\beta'}{\alpha'}
= 
\frac{1+ \frac{\lambda C'_1 }{2m(1-\frac{ C'_2 \lambda}{m })}}{1- \frac{\lambda C'_1 }{2m(1-\frac{ C'_2 \lambda}{m})}}
 \leq
  \frac{1+ \frac{\lambda C'_1 }{m} }{1- \frac{\lambda C'_1 }{m} }.
$$
Then, $ \lambda \leq \frac{m \delta}{C'_1(1+\delta)} $ leads to
$
\frac{\beta'}{\alpha'} \leq (1+\delta).
$

Note that $\lambda'^* = m \min(1/(2C'_2), \delta/[C'_1(1+\delta)] )$ satisfies these two conditions, so from now $\lambda = \lambda'^*$. We also use the following:
\begin{align*}
\frac{1}{\alpha'} 
 = 
\frac{1}{\lambda'^*\left(1- \frac{\lambda'^* C'_1 }{2m( 1- C_2 \lambda'^* / m )}\right)}
  \leq 
\frac{\beta'}{\lambda'^* \alpha'} 
 \leq 
\frac{(1+\delta)}{m \min(1/(2C'_2), \delta/[C'_1(1+\delta)] )} 
 \leq 
\frac{C'_1 (1+\delta)^2 }{m \delta}.
\end{align*}
So far the bound is: 
\begin{multline*}
\mathbb{P}\Biggl\{  \left\|
\widehat{XMZ}_{\lambda'^*} -XM^*Z \right\|_{F,\Pi}^2
\leq 
\inf_{\bar{M}\in\mathbb{R}^{p\times k} } \Biggl[ (1+\delta) \|X\bar{M}Z - X M^*Z \|_{F,\Pi}^2 
+
 (1+\delta) \|X \|_{F}^2\|Z \|_{F}^2 p k\tau^2
\\
+ 
 \frac{C'_1 (1+\delta)^2 \left( 4 {\rm rank} (\bar{M}) (k+p+2) \log \left( 1+ \frac{\| \bar{ M } \|_F}{\tau \sqrt{2{\rm rank} (\bar{M})}} \right) + 2 \log \frac{2}{\varepsilon} \right)}{ m \delta } \Biggr] \Biggr\} 
\geq 
1-\varepsilon.
\end{multline*}
In particular, with probability at least $ 1-\varepsilon $, the choice $\tau^2 = C'_1 (k+p)/(m k p \|X \|_{F}^2 \|Z \|_{F}^2 )$ gives
\begin{multline*}
 \left\|\widehat{XMZ}_{\lambda'^*} -XM^*Z \right\|_{F,\Pi}^2
\leq 
\inf_{\bar{M}\in\mathbb{R}^{p\times k} } \Biggl[ (1+\delta) \|X\bar{M}Z - X M^*Z \|_{F,\Pi}^2 + \frac{C'_1 (1+\delta) (k+p)}{m }+
\\
 \frac{C'_1 (1+\delta)^2 \left( 4 {\rm rank} (\bar{M}) (k+p+2) \log \left( 1+\frac{\| X \|_F \|Z \|_{F} \| \bar{ M } \|_F}{ \sqrt{C'_1}} 
 	\sqrt{ \frac{m k p}{(k+p) {\rm rank}(\bar{M})}} \right) + 2 \log \frac{2}{\varepsilon} \right)}{ m\delta } \Biggr].
\end{multline*}
\end{proof}

\subsection{Proof of theorem \ref{thrm_contraction}}

\begin{proof}
	The proof is proceeded completely similar to the proof of Theorem \ref{thrm_contraction_blr}, in Section \ref{sc_proof_thrm_contraction_blr}.
		
\end{proof}

\section{Comments on algorithm implementation}

For the case of inductive matrix completion, we write the logarithm of the density of the posterior
$$
\log  \widehat{\rho}_{\lambda}(M)
= 
- \frac{\lambda}{n} \sum_{i=1}^n ( Y_i - (\Pi_C(XMZ))_{i} )^2
- 
\frac{p+m+2}{2} \log \det (\tau^2 \mathbf{I}_{m} + MM^\intercal ) .
$$
Let us now differentiate this expression in $M$. Note that the term $( Y_i - (\Pi_C(XMZ))_{i} )^2$ does actually not depend on $M$ locally if $(XMZ)_{i}\notin [-C,C]$, in this case its differential with respect to $M$ is $\mathbf{0}_{p \times m}$. Otherwise, $( Y_i - (\Pi_C(XMZ))_{i} )^2=( Y_i - (XMZ)_{i} )^2$.
In order to be able to differentiate the term $(XMZ)_{i}$, let us introduce a notation for the entries of $\mathcal{I}_i$: $\mathcal{I}_i =(a_i,b_i)$.
Then
$ \nabla (XMZ)_{i} =D $ where the matrix $D\in\mathbb{R}^{p\times m}$ satisfies $D_{x,y} = \mathbf{1}_{\{x=b_j\}} X_{a_j,y}$. Then
$$
\nabla  \log   \widehat{\rho}_{\lambda}( M)
= 
\frac{2\lambda}{n} \sum_{i=1}^n (\nabla (XMZ)_{i}) ( Y_i - (XMZ)_{i} ) \mathbf{1}_{\{ |(XMZ)_{i}| < C \}}
- 
(p+m+2) (\tau^2 \mathbf{I}_{m} + MM^\intercal )^{-1} M .
$$

\noindent In the above calculation still requires to calculate a $p\times p$ matrix inversion at each iteration, for very large $ p $, this might be expensive and can slow down the algorithm. Therefore, we could replace this matrix inversion by its accurately approximation through a convex optimization.  It is noted that the matrix
$\mathbf{B} := (\tau^2 \mathbf{I}_{m} + MM^\intercal )^{-1} M $ is the solution to the following convex optimization problem: 
$
\min_{\mathbf{B} } \big\{\| \mathbf{I}_p- M^\top \mathbf{B}  \|_F^2 + \tau^2\|\mathbf{B} \|_F^2\big\}.
$
The solution of this optimization problem can be conveniently obtained by using the package `\texttt{glmnet}' \cite{glmnet} (with the family option `\texttt{mgaussian}'). This avoids to perform matrix inversion or other costly calculation.  However,  we note here that the LMC algorithm is being used with approximate gradient evaluation, theoretical assessment of this approach can be found in \cite{dalalyan2019user}.

\begin{algorithm}
	\caption{LMC}
	\begin{algorithmic}[0]
		\State \textbf{Input}: The data.
		\State \textbf{Parameters}: Positive real numbers $ \lambda, \tau,h,T $. 
		\State \textbf{Output}: The matrix $\widehat{M} $
		\State \textbf{Initialize}: $ M_0 ,  \widehat{M} = \mathbf{0}_{m\times p}$ 
		\For{$k \gets 1$ to $T$} 
		\State Sample $ M_{k} $ from \eqref{langevinMC};
		\State  $ \widehat{M} \gets  \widehat{M} +  M_{k}/T $
		\EndFor
	\end{algorithmic}
	\label{lmc_algorithm}
\end{algorithm}

\begin{algorithm}{}
	\caption{MALA}
	\begin{algorithmic}[0]
		\State \textbf{Input}: The data.
		\State \textbf{Parameters}: Positive real numbers $ \lambda, \tau,h,T $
		\State \textbf{Output}: The matrix $\widehat{M} $
		\State \textbf{Initialize}: $M_0 ;  \widehat{M} = \mathbf{0}_{m\times p}$ 
		\For{$k = 1$ to $T$} 
		\State Sample $\tilde{M}_{k} $ from \eqref{mala}.
		\State Set $M_k = \tilde{M}_{k}$ with probability $ A_{MALA} $, from \eqref{eq_ac_mala}, otherwise $M_k = M_{k-1}$ .
		\State  $\widehat{M} \gets  \widehat{M} +  M_{k}/T $ .
		\EndFor
	\end{algorithmic}
	\label{mala_algoritm}
\end{algorithm}

\clearpage
\begin{footnotesize}
\bibliographystyle{abbrv}

\end{footnotesize}

\end{document}